\documentclass[
times,
aps,
twocolumn,
showpacs,
superscriptaddress,
pre,
floatfix]{revtex4-2}

\usepackage[centering,hmargin=1.7cm,vmargin=0.78in]{geometry}

\usepackage{amsmath,mathtools,amsthm,amssymb,pifont}
\usepackage[utf8]{inputenc}
\usepackage[T1]{fontenc}
\usepackage[american]{babel}
\usepackage{graphicx,bbold,titlesec}
\usepackage{braket}
\usepackage{float}
\usepackage{color}
\usepackage{hyperref}

\usepackage{pifont}
\usepackage{bbold}
\usepackage{biolinum}

\usepackage{mathtools}
\usepackage{bm}

\usepackage{array}
\usepackage{color}
\usepackage{upgreek}
\usepackage{float}
\usepackage{booktabs}
\usepackage{url}
\usepackage{braket}
\usepackage{amsthm}
\usepackage{bm}% bold math
\usepackage[T1]{fontenc}
\usepackage{scrextend}
\usepackage{enumerate}
\usepackage{enumitem}
\usepackage[dvipsnames]{xcolor}
\definecolor{BlueRome}{HTML}{4287f5}
\definecolor{C1}{RGB}{52, 89, 149}
\definecolor{C2}{RGB}{251, 77, 61}
\definecolor{C3}{RGB}{3, 206, 164}
\definecolor{C4}{RGB}{202, 21, 81}
\definecolor{cyanRed}{RGB}{252, 5, 120}
\hypersetup{colorlinks=true, linkcolor=BlueRome, citecolor=BlueRome, urlcolor=cyanRed}

\definecolor{purple}{HTML}{E5E3F5}
\definecolor{amaranth}{HTML}{F5CFD8} % liek a red
\definecolor{yellow}{HTML}{FFEAB6}
\definecolor{pink}{HTML}{FFDBDA}

\definecolor{darkpurple}{HTML}{3D348B}
\definecolor{darkamaranth}{HTML}{AB2346}
\definecolor{darkyellow}{HTML}{FFBA08}
\usepackage{dsfont}
\usepackage{epstopdf}
\usepackage{tikz}
\usepackage{mathrsfs}
\usepackage[export]{adjustbox}
\usepackage{thmtools}
\usepackage{amsmath}
\usepackage{thm-restate}
\usepackage{enumerate}
\usepackage{enumitem}
\usepackage{amssymb}
\newtheorem{thm}{Theorem}

\newtheorem{lemma}[thm]{Lemma}

\newtheorem{mydef}{Definition}

\theoremstyle{remark}

\newcommand*{\ot}{\otimes}
\newcommand*{\nn}{\nonumber}

\newcommand*{\id}{\mathds{1}}

\newcommand*{\llangle}{\langle \! \langle}
\newcommand*{\rrangle}{\rangle \! \rangle}

\newcommand*{\mc}{\mathcal}

\DeclareMathOperator{\tr}{tr}

\usepackage{mathtools}
\usepackage{thm-restate}

\usepackage[table]{xcolor}

\newenvironment{hproof}{%
  \proof}{\endproof}

\usepackage[export]{adjustbox}
\hypersetup{
pdftitle={Classical Simulability from Operator Entanglement Scaling},
pdfsubject={tensor networks},
pdfauthor={Neil Dowling},
pdfkeywords={Quantum chaos, Many-body quantum physics, tensor networks}
}

\begin{document}

\title[]{{Classical Simulability from Operator Entanglement Scaling}}

\author{Neil Dowling}
\email[]{ndowling@uni-koeln.de}
\affiliation{Institut f\"ur Theoretische Physik, Universit\"at zu K\"oln, Z\"ulpicher Strasse 77, 50937 K\"oln, Germany}

% \date{\today}
% \pacs{}

% <600 characters - = 599 at the moment!!! (w/o spaces)
\begin{abstract}
     Local-operator entanglement (LOE) quantifies the nonlocal structure of Heisenberg operators and serves as a diagnostic of many-body chaos. We provide rigorous bounds showing when an operator can be well-approximated by a matrix-product operator (MPO), given asymptotic scaling of its LOE $\alpha$-R\'enyi entropies. Specifically, we prove that a volume law scaling for $\alpha\geq 1$ implies that the operator cannot be approximated efficiently as an MPO while faithfully reproducing all expectation values. On the other hand, if we restrict to correlations over a relevant sub-class of (ensembles of) states, then logarithmic scaling of the $\alpha < 1$ entropies implies MPO simulability.
     This result covers a range of relevant quantities, including infinite temperature autocorrelation functions, out-of-time-ordered correlators, and average-case expectation values over ensembles of computational basis states. Beyond this regime, we provide numerical evidence together with a random matrix model 
    to argue that this simulability result also typically holds for arbitrary states.
    % also for out-of-equilibrium expectation values, logarithmic scaling for $\alpha < 1$ R\'enyi LOE typically guarantees simulability.    
    Our results put on firm footing the heuristic expectation that a low operator entanglement implies efficient tensor network representability, extending celebrated foundational results from the theory of matrix-product states and providing a formal link between quantum chaos and classical simulability. 
\end{abstract}


\maketitle
% \com{LENGTH  Limit is 3750}

% \stoptoc
\section{Introduction}
% \textit{Introduction.---} 
Tensor networks play a central role in the numerical study of quantum many-body systems. Their efficiency relies on the structure of entanglement: in one dimension, a state admits a faithful matrix-product state (MPS) representation with polynomial cost in the system size $N$ when its entanglement entropy scales at most logarithmically with $N$~\cite{Verstraete2006}. This principle underpins the remarkable success of tensor network-based algorithms~\cite{Eisert_Cramer_Plenio_2010,ORUS2014117} for tasks such as determining the ground state of local Hamiltonians~\cite{WhiteDMRG,SCHOLLWOCK201196} and simulating dynamics~\cite{PhysRevLett.93.040502,PhysRevLett.107.070601}.
However, in the case of time evolution, one quickly encounters the so-called `entanglement barrier': the entanglement of an out-of-equilibrium many-body state tends to grow as a volume-law when evolving according to non-localized dynamics~\cite{Calabrese_2005}, implying that tensor network-based methods are fundamentally unsuitable for studying the dynamics of the full many-body wavefunction, at least directly. 

Shifting perspective, one can instead apply the principles of tensor networks to \textit{operator} evolution. The relevant question is then how well a Heisenberg operator can be represented as a matrix-product operator (MPO)~\cite{Prosen2007,Hartmann2009}; see Fig.~\ref{fig:main-results}. Although physically equivalent to evolving a state, simulating the relevant operator instead can lead to an exponential computational advantage. For instance, a local Pauli $\sigma_z$ evolving under a Gaussian Fermionic Hamiltonian can be represented exactly as an MPO of constant bond dimension~\cite{Prosen2007,Muth2011}, in contrast to the volume-law growth of entanglement of these systems in the Schr\"odinger picture~\cite{Calabrese_2005}. Beyond this, intrinsic properties of operators---inaccessible via sampling a single quantum state---are also of interest: infinite temperature autocorrelation functions (ITACs) govern linear response theory~\cite{richter_impact_2019,Richter2019} and 
% diffusive
operator hydrodynamics~\cite{Steinigeweg2009,Bertini2021transport,Wang2024}, while out-of-time-ordered correlators (OTOCs) probe information scrambling~\cite{Shenker_Stanford_2014,Nahum2018,dowling2023scrambling}.

\begin{figure}[t]
    \centering
    \includegraphics[width=\linewidth]{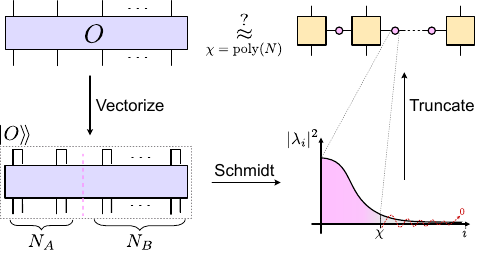}
    \caption{A schematic of the setting studied in this work. Anti-clockwise from the top-left: An operator $O$ on a chain of $N=N_A+N_B$ qudits is in one-to-one correspondence to a pure state $| O \rrangle$ in a doubled Hilbert space. Schmidt decomposition across some bipartition leads to a spectrum of singular values $\{ |\lambda_i|^2 \}_{i=1}^{d^{2N_A}}$, from which the (local-)operator entanglement (LOE) entropies are defined. Truncating and keeping only terms with the largest $\chi$ Schmidt values across each bipartition leads to an MPO approximation of $O$. We study when a sufficiently small [large] LOE implies the [non-]\! existence of an efficient and faithful MPO approximation.
    }
    \label{fig:main-results}
\end{figure}

If the scaling of entanglement dictates the efficiency of MPS methods, what governs the cost of simulating an {operator} as an MPO? 
An analogous quantity to `entanglement' can be defined for operators. Namely, a suitably normalized operator is in one-to-one correspondence with a pure state through its action on half of a maximally entangled state across two copies of Hilbert space, 
\begin{equation}
    | O \rrangle := (O \otimes \id) \ket{\phi^+}. \label{eq:choi}
\end{equation}
The entanglement of $| O \rrangle$ across some spatial bipartition is called the (local-)operator entanglement (LOE)~\cite{Zanardi2001,Prosen2007}, studied as a witness of many-body quantum chaos~\cite{Prosen2007,Prosen2007a,Prosen2009,Dubail_2017,Jonay2018,Alba2019,Kos2020,Kos2020II,Alba2021,dowling2023scrambling,Dowling2024butterfly,Alba2025-mx,bertini2025randomperm}. We call it `LOE' here to be consistent with this literature and to emphasize that we are working with Heisenberg operators, but make no assumptions on the locality of the operator~\footnote{LOE should be distinguished from the `operator-space entanglement entropies' of the unitary evolution propagator~\cite{Zanardi2001}, observed to always scale as a volume-law for non-localized dynamics~\cite{Zhou2017,Dubail_2017}, and so it is not as relevant to the present questions on efficient simulation.}. Relevant for us, the LOE is exactly the entropy of the Schmidt coefficients of the corresponding bond when writing $O$ as an MPO.

Drawing an analogy to the case of MPS, one might expect that the scaling of its LOE entropies is indicative of a Heisenberg operator's representability as an MPO. Indeed, this is often informally assumed~\cite{Prosen2007,swingle2020,Rakovszky2022,begusic2024realtime,Dowlin2024LOE-OSRE}. However, unlike for the theory of MPS~\cite{Verstraete2006,Schuch_2008,Eisert_Cramer_Plenio_2010,ORUS2014117}, so far these claims lack a rigorous foundation. 

Here, we close this gap. We provide theoretical guarantees on the simulability of Heisenberg operators from the scaling of their LOE R\'enyi entropies in 1D. Whether a given operator can be efficiently represented as an MPO depends intrinsically on the state over which one wants to faithfully reproduce correlations. We first prove that a linear growth of  $\alpha\geq 1$ LOE entropies necessarily indicates that no MPO can efficiently reproduce expectation values for \textit{all} states. On the other side, we show that a logarithmic growth of $\alpha<1$ R\'enyi entropies implies efficient MPO approximability according to correlations over a relevant class of states (such as ITACs),
and extend this to higher-order OTOCs. These results are summarized in Table~\ref{tab:results}. Finally, via numerical evidence and a random matrix model, we argue that in practice, for physically relevant operators, the efficient simulability results should carry over to arbitrary states.

\section{Operator Entanglement and Matrix Product Operators}
% \textit{LOE and MPOs.---} 
Consider an $N$-qudit spin-chain Hilbert space $\mc{H}$ of dimension $D=d^N$, and some Hermitian operator $O$ on this space. We consider open or fixed boundary conditions, but the following generalizes directly to periodic boundaries. We take $O$ to be `normalized' and `bounded', by which we mean the Hilbert-Schmidt norm is given by $\| O \|_2=\sqrt{\tr[O^\dagger O]}=\sqrt{D}$ and that the spectral norm satisfies $\| O \|_{\infty} \leq 1$ respectively~\footnote{$\| O\|_p := \tr[|O|^p]^{1/p}$ refers to the Schatten $p$-norm of $O$, with $\| O \|_{\infty} $ equal to its largest singular value.}. 
In this case, $O$ is isomorphic to a pure state through the vectorization/normalization mapping, as shown in Eq.~\eqref{eq:choi} and Fig.~\ref{fig:main-results}. Using this mapping, the LOE R\'enyi entropies of $O$ across some spatial bipartition in (doubled) Hilbert space, $\mc{H}_A^{\ot 2} \otimes \mc{H}_{B}^{\ot 2}$, are defined as
\begin{equation}
    E^{(\alpha)}_A(O) := (1-\alpha)^{-1} \log( \tr[\tr_{B}[|O\rrangle\!\llangle O|]^{\alpha}] ) \label{eq:loe}
\end{equation}
for $\alpha \geq 0$, and $\alpha=0,1,\infty $ defined through appropriate limits. We stress that the partial trace is over local sites of dimension $d^2$; see App.~\ref{ap:proofs}. Throughout, we use base $2$ logarithms and so entropies are measured in bits. 

\begin{figure}
\centering
\renewcommand{\figurename}{Table}
\setcounter{figure}{0}  
\includegraphics[width=\linewidth]{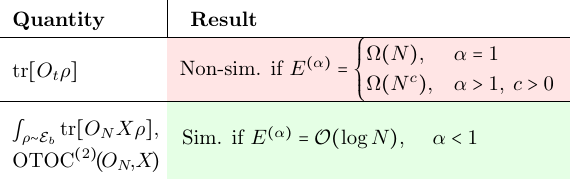}
\caption{Results on the approximability of an operator's properties from the scaling of its LOE R\'enyi entropies, Eq.~\eqref{eq:loe}. Simulability (`sim.') is defined in Def.~\ref{def:sim}, $\mc{E}_b$ refers to ensembles of states with a sufficiently mixed first moment (Def.~\ref{def:low-av}), $X$ is some other normalized and bounded Hermitian operator, and OTOCs are defined in Eq.~\eqref{eq:otocDef}. The top row summarizes Thm.~\ref{thm:micro} and the bottom Thms.~\ref{thm:itac} and~\ref{thm:otoc}.}
    \label{tab:results}
\end{figure}

Through Eq.~\eqref{eq:choi}, one can adapt standard many-body state techniques to operator space. Schmidt decomposition of $|O \rrangle$ allows us to truncate across a bipartition $\mc{H}=\mc{H}_A \otimes \mc{H}_B$, to obtain a rank $\chi$ approximation to $O$,
\begin{equation}
    O = \sum_{i=1}^{d^{2 n}} \lambda_i A_i \otimes B_i \underset{\mathrm{truncate}}{\to} \sum_{i=1}^{\chi} \lambda_i A_i \otimes B_i =: \tilde{O}_{\chi}^{(n)}.\label{eq:schmidt}
\end{equation}
Here, $\lambda_i \geq \lambda_{i+1}$ with $\sum_i |\lambda_i|^2=1$, $\tr[A_i^\dagger A_j]=D_A\delta_{ij}$, and $  \tr[B_i^\dagger B_j]=D_B\delta_{ij}$ where $A_i$ has support on the first $n=N_A$ sites (as indicated by the $n$ superscript), and WLOG we take $N_A \leq N_B$ throughout this work. A truncation strategy can be successively applied for every spatial bipartition (each $n=1,2,\dots,N-1$), using standard MPS methods~\cite{ORUS2014117} to arrive at an MPO approximation to $O$,
\begin{align}
    \tilde{O}_{\chi} := \sum_{i_1,\dots,i_N=0}^{d^2-1} \tr[\Lambda^{(1)}_{i_1} \dots \Lambda^{(N)}_{i_N} ] P_{i_1} \otimes \dots \otimes P_{i_N}. \label{eq:mpo}
\end{align}
Above, $\Lambda^{(\ell)}_{i_\ell} $ are matrices of dimension  $\chi \times \chi$ and $P_{i_\ell}$ are some orthonormal operator basis for the spin $\ell$: $\tr[P_{i}^\dagger P_j] = d \delta_{ij}$. Natural choices are the computational basis matrix elements (with a $\sqrt{d}$ normalization)~\cite{Hartmann2009}, 
or Pauli matrices, which have the advantage of being amenable to Pauli propagation principles~\cite{Rakovszky2022,begusic2024realtime,schuster2024polynomialtime,angrisani2024classically,dowling2024magicheis,Dowlin2024LOE-OSRE,rudolph2025paulipropag}. Iterative evolution then truncation is the principle behind Heisenberg picture density matrix renormalization group algorithms (H-DMRG) for time evolution~\cite{Hartmann2009}. Throughout, we use $\tilde{O}_{N,\chi}^{(n)}$ to indicate an approximation to $O_N$ obtained through a single rank $\chi$ truncation of $O_N$ between the first $n$ qudits and the rest, while $\tilde{O}_{N,\chi}$ refers to a rank $\chi$ MPO approximation to $O_N$.

In the remainder of this work we study when an operator can be faithfully and efficiently represented as an MPO compared to the scaling of its LOE entropies. To rigorously present our results, we first define precisely what we mean by `efficient approximation'.
\begin{mydef} \textbf{(MPO Simulability)}
    Consider a discrete family of Hermitian operators $O_N$ enumerated by $N$, its qudit support, together with a set of functionals of these operators $\mathcal{F}=\{ f_N\}$, where each $f_N: \mathrm{Herm}_{d^N} \to \mathbb{R}$. If, for any $\varepsilon >0$, there exists a corresponding family of MPOs $\tilde{O}_{N,\chi}$ with bond dimension $\chi = \mathrm{poly}(N,\varepsilon^{-1} ) $ such that for all $f_N \in \mathcal{F}$: 
    \begin{equation}
        |f_N(O_N)-f_N(\tilde{O}_{N,\chi}) | \leq \varepsilon,
    \end{equation}
    then we say that $O_N$ is {efficiently approximable as an MPO} according to the measure induced by $\mathcal{F}$. \label{def:sim}
\end{mydef}
This definition is adapted from that found in Refs.~\cite{Verstraete2006,Schuch_2008}, with the subtle but important difference being the explicit dependence on restricted quantities defined by $\mathcal{F}$. For instance, $\mathcal{F}$ could mean any expectation values according to states of a certain temperature, or all OTOCs for normalized probe operators. This definition will help uncover how certain physically relevant properties of an operator---i.e., particular functionals of it---can be much cheaper to approximate compared to others. In the following, we will use `operator' or `state' to refer to families thereof, well-defined on $N$ sites. We will use MPO ``simulability'', ``approximability'' and ``representability'' to all mean the same thing, as formally described in Def.~\ref{def:sim}. Given the existence of an efficient MPO representation, one expects standard techniques to efficiently arrive there (such as H-DMRG~\cite{Eisert_Cramer_Plenio_2010}). Guarantees for particular algorithms to be efficient are a stronger sense of simulability, and would likely require more stringent, physical assumptions than the general setting adopted here. We leave further analysis of concrete algorithms to future work.

\section{Simulability Guarantees}
% \textit{Simulability Guarantees.---} 
A natural first example are expectation values of $O_N$ according to some (class of) states $\rho$.
From H\"older's inequality, we have that for any state $\rho$, the error in correlations can be tightly bound as
\begin{equation}
    |\tr[(O_N-\tilde{O}_{N,\chi}) \rho]| \leq \| \rho \|_1  \| \Delta O \|_\infty = \| \Delta O \|_\infty,\label{eq:upperHold}
\end{equation}
where $\Delta O := O_N-\tilde{O}_{N,\chi}$.
From Eq.~\eqref{eq:upperHold}, we arrive at a no-go theorem when the LOE obeys a volume-law.
\begin{thm} \label{thm:micro}
    Consider a Hermitian operator $O_N$ and an MPO approximation $\tilde{O}_{N,\chi}$. Then there exists a state $\rho$ such that the error in its expectation value, $\varepsilon := |\tr[\Delta O \rho ]|$, for $\alpha>1$ satisfies
    \begin{equation}
            \log(\chi) \! \geq  \!\max \Big\{ 
            \! E^{(1)} -\varepsilon N\log(d) -1,E^{(\alpha)}+ \frac{\alpha \log(1-\varepsilon^2 )}{\alpha - 1} \! \Big\} , \nn
            % \label{eq:thm1}
    \end{equation}
    where $E^{(\alpha)}:=\max_{A}\{ E^{(\alpha)}_A(O_N)\}$ is the maximum LOE over contiguous bipartitions (i.e., over choice of $n$).
\end{thm}
\begin{hproof}
    Given that $O_N$ is Hermitian, we can choose a Hermitian, local basis in Eq.~\eqref{eq:mpo} that ensures that, after singular value decomposition of the correlation matrix, $\tilde{O}_{N,\chi}$ is also Hermitian. Then we know there exists a state $\rho$ for which the error is equal to the spectral norm distance: $\varepsilon :=|\tr[\Delta O \rho]|=\| \Delta O \|_{\infty}$. Now, it is sufficient to consider any bipartition $\mc{H}_A \otimes \mc{H}_B$ that leads to a rank-$\chi$ approximation through truncation. Through a Schatten norm inequality and the normalization of $O_N$, 
    \begin{equation}
        \Big({\sum_{i=\chi+1}^{D_A^2} \lambda_i^2 }\Big)^{1/2}= D^{-1/2}\| \Delta O^{(n)}\|_2 \leq  \varepsilon, \label{eq:lower}
    \end{equation}
    where $\Delta O^{(n)}:= O_N-\tilde{O}_{N,\chi}^{(n)} $ and $\{\lambda_i\}_{i=1}^{d^{2n}}$ are the Schmidt values of $|O\rrangle$ across the specified cut (cf. $\|\Delta O\|_\infty$, which corresponds to the error compared to a full MPO approximation). In the final inequality, we have used that $\| \Delta O^{(n)}\|_2 \leq \| \Delta O \|_2 \leq D^{1/2} \|  \Delta O \|_\infty=:\epsilon$. We then proceed by applying the two main results of Ref.~\cite{Schuch_2008} to the states $|O_N\rrangle$ and $|\tilde{O}_{N,\chi}\rrangle$, to lower-bound the left-hand side of Eq.~\eqref{eq:lower} in terms of the entropies of the distribution $\{ | \lambda_i|^2\} $. As the above argument holds for any bipartition, we can pick that which leads to the tightest bound on $\chi$. A detailed proof is supplied in App.~\ref{ap:nogo}.
\end{hproof}
Considering the negation of Def.~\ref{def:sim} and recalling Landau asymptotic notation~\footnote{$f(N)= \mc{O}(g(N))$ and $f(N)=\Omega(g(N)) $ mean, respectively, that $\lim_{N \to \infty} f/g \leq c$ and $\lim_{N \to \infty} f/g \geq c$, for some constant $c$.}, Thm.~\ref{thm:micro} means that if either $E_A^{(1)}(O_N) = \Omega(N)$ or $E_A^{(\alpha)}(O_N) = \Omega(N^c)$ for $\alpha > 1$ and $c>0$ across a bipartition, then $\chi$ is superpolynomial in $N$ and so $O_N $ cannot be efficiently approximated as an MPO according to expectation values for all quantum states. This is summarized in the first row of Table~\ref{tab:results}.

A similar relation to Thm.~\ref{thm:micro} holds between MPS and state R\'enyi entanglement entropies~\cite{Schuch_2008}. However, if we now instead ask about \textit{approximability} of an operator as an MPO, we begin to see deviations from MPS theory. Eq.~\eqref{eq:upperHold} means that the spectral norm bounds arbitrary expectation values. In contrast, operator truncation (Eq.~\eqref{eq:schmidt}) is optimal in the Hilbert-Schmidt norm error---which in turn can be related to LOE entropies.
In full generality, their only relation is: $D^{-1/2}\| \Delta O\|_2 \leq \| \Delta O \|_\infty \leq \| \Delta O\|_2$. 
% Therefore, $ \| \Delta O \|_\infty$ can be exponentially larger than $D^{-1/2}\| \Delta O \|_2$. 
Circumventing the worst-case scenario of Thm.~\ref{thm:micro}, our strategy will now be to restrict to relevant classes of states for which $O_N$ is measured. The simplest example is infinite temperature correlations, 
\begin{equation}
    |{D^{-1}}\tr[\Delta O  X] | \leq D^{-1/2}{\| \Delta O \|_2} ,\label{eq:itac}
\end{equation}
from the Cauchy-Schwarz inequality for any normalized $X$ ($\tr[X^\dagger X]=D$). Beyond infinite temperature correlations, we consider ensembles of states whose average is sufficiently mixed. 
\begin{mydef} \textbf{(Low-average Ensemble)} \label{def:low-av}
    An ensemble $\mc{E}_b$ of quantum states is called low-average if its first moment satisfies $\| \overline{\rho} \|_\infty \leq {b}{D^{-1}}$ for some $ b \geq 1$.
\end{mydef}
This definition follows closely from one introduced in Ref.~\cite{schuster2024polynomialtime}, though we allow $b$ to scale with $N$. Relevant examples include: the trivial ensemble of a single, sufficiently mixed density matrix (cf. Eq.~\eqref{eq:itac}), uniform distributions of computational basis states, and unitary designs. For these states, we can relate MPO simulability to the LOE, extending Eq.~\eqref{eq:itac}. 
\begin{thm} \label{thm:itac}
     Consider an operator $O_N$ and a low-average ensemble of states $\mc{E}_b$. Then there exists an MPO approximation $\tilde{O}_{N,\chi}$ such that on average over $\rho \sim \mc{E}_b$ and 
     % error in approximating $O_N$ as an MPO 
     for $0<\alpha<1$, 
    \begin{equation}
         \overline{| \tr[ \Delta O X \rho ] | }\leq  \!N b^{1/2} \exp_2{\!\left( \! \frac{1-\alpha}{2\alpha}\Big( E^{(\alpha)} \! -\! \log\big( \frac{\chi}{1-\alpha}  \!\big)  \!\Big)  \!\right) \!}, \! \nn
         % \label{eq:thm2}
    \end{equation}
    where $E^{(\alpha)}:=\max_{A}\{ E^{(\alpha)}_A(O_N)\}$, $X$ is some bounded and normalized operator, and $\exp_2(x) :=2^x$.
\end{thm}
\begin{hproof}
    We first upper-bound the average $\overline{| \tr[ \Delta O X \rho ] | }$ by the normalized Hilbert-Schmidt distance, $D^{-1/2} {\| \Delta O \|_2}$, using a series of elementary inequalities. Then, we show that $\| \Delta O \|_2$ is bound from above by the sum of errors from a single-cut truncation across each bond of the MPO, $\| \Delta O^{(n)}\|_2$, through a series of triangle inequalities and via the contractivity of the Hilbert-Schmidt norm under a single truncation (which is an orthogonal projection). We recall from Eq.~\eqref{eq:lower} that the two-norm distance for a single truncation is: $D^{-1/2}\|\Delta O\|_2=
        ({\sum_{i=\chi+1}^{D_A^2} \lambda_i^2 })^{1/2}$. We can therefore apply an analogue of the tail-sum result used in Thm.~\ref{thm:micro} to relate $\sum_{i=\chi+1}^{D_A^2} \lambda_i^2$ to the $\alpha <1$ LOE for each cut, and finally we `stitch-together' the approximations for each cut through a sequence of triangle inequalities. The full proof uses results from Refs.~\cite{Verstraete2006,Schuch2020,schuster2024polynomialtime} and can be found in App.~\ref{ap:thm2}.
\end{hproof}
Assuming that $E^{(\alpha)}_A(O_N) \leq c \log(N)$ for $\alpha <1$ and $c>0$ for any bipartition, this result implies that $|\tr[\Delta O X \rho ]| \leq \varepsilon$ is satisfied on-average whenever
\begin{equation}
    \chi \geq N^c(1-\alpha)\left( {b N^2}{\varepsilon^{-2}} \right)^{\frac{\alpha}{1-\alpha}}.\label{eq:itac2}
\end{equation}
Therefore, for $E^{(\alpha)}_A(O_N) =\mc{O}(\log(N))$, $O_N$ is efficiently approximable as an MPO on average when $b  = \mc{O} (\mathrm{poly}(N)) $, as summarized in Table~\ref{tab:results}.

Note that only the reduced state $\rho$ on the support of $O_N$ should be considered: if the relevant state is highly entangled between the support of $O_N$ and its complement (i.e. the region outside its light-cone for a time-evolving local operator), then this is equivalent to computing expectation values over a highly mixed state and so Thm.~\ref{thm:itac} is relevant in this case; cf. Ref.~\cite{xu2026classicalsimulationnoiselessquantum}. 

% \textit{Examples.---} 
\section{Examples}
Before moving on to generalizations for higher-order correlators, we consider some concrete examples. To compare Thms.~\ref{thm:micro} and \ref{thm:itac}, consider the Hermitian, bounded, and normalized $O_N :=\id - 2  \sum_{i=1}^r |\psi_i\rangle \! \langle \psi_i |$, where we take the orthonormal $\ket{\psi_i}$ to be product for $i \leq \chi+2$ and generic (highly entangled) for $i> \chi+2$. If we truncate to rank $\chi$, we can find an approximation with $\| \Delta O \|_\infty =2$, but $D^{-1/2} \| \Delta O \|_2 \approx (r/D)^{1/2} = \mc{O}(N^{-1})$ for a suitably chosen $r$. Such an operator is therefore efficiently approximable as an MPO for correlations over low-average ensembles, but not according to all possible states.    

There are several physical systems where the logarithmic scaling of LOE is known exactly, and so for which Thm.~\ref{thm:itac} directly applies. In these models, time serves as a proxy for $N$: an initially local operator evolving according to a local circuit or Hamiltonian leads to a family of (time-evolved) operators $O_N$ with support on $N\approx 2t$ sites from its light-cone. Examples include the integrable models:  
(i) free Fermion hopping dynamics~\cite{Prosen2007a,Hartmann2009,Dubail_2017}, (ii) dual unitary XXZ circuit model~\cite{Kos2020II}, and (iii) the Rule 54 circuit model~\cite{Alba2019,jacoby2025longtime}. 
It is conjectured that the characteristic logarithmic growth of LOE persists in all locally-interacting integrable models~\cite{Prosen2007}, with numerical evidence further supporting this~\cite{Prosen2007,Prosen2009,Alba2021,Alba2025-mx}. Our results, therefore, establish a precise relation between many-body chaos and classical simulability. In fact, this relationship goes beyond Bethe-ansatz integrability: Clifford circuits with only $\mc{O}(\log(N))$ non-Clifford gates (`magic' resources) also lead to a logarithmically bounded LOE~\cite{Dowlin2024LOE-OSRE}.

\section{Higher-order Correlators}
% \textit{Higher-order Correlators.---} 
We now consider the approximability of non-linear properties of $O_N$. Higher-order OTOCs are defined for $k \geq 2$ as 
\begin{equation}
    \mathrm{OTOC}^{(k)}(O,X):=D^{-1}\tr[(OX)^k].\label{eq:otocDef}
\end{equation}
The case of $k=2$ has been extensively studied in the context of information scrambling~\cite{Shenker_Stanford_2014,Nahum2018,Swingle2018,swingle2020,dowling2023scrambling}, while $k>2$ OTOCs are recently of interest regarding their relation to 
% Free Probability and 
the foundations of quantum statistical mechanics~\cite{Foini2019,Pappalardi2022freeETH,fritzsch2025freep,vallini2025refinementseig}, unitary designs~\cite{Roberts2017-en,fava2023designsfreeprobability,dowling2025freeindep}, and quantum supremacy experiments~\cite{Mi2021,Abanin2025}. We now extend Thm.~\ref{thm:itac} to OTOCs.
\begin{thm} \label{thm:otoc}
     Consider an operator $O_N$. Then there exists an MPO approximation $\tilde{O}_{N,\chi}$ such that for $0<\alpha<1$,
    \begin{align}
        | \Delta \mathrm{OTOC}^{(k)} |\leq &N \left(  \frac{\|\tilde{O}_{N,\chi}\|_\infty^{k-1}-1}{\|\tilde{O}_{N,\chi}\|_\infty-1} + 1 \right) \nn \\
        &\times \exp_2\left( \frac{1-\alpha}{2\alpha}\Big( E^{(\alpha)} - \log\big( \frac{\chi}{1-\alpha} \big) \Big) \right),  \nn
    \end{align}
        where $E^{(\alpha)}:=\max_{A} \{ E^{(\alpha)}_A(O_N)\}$, and
        $\Delta \mathrm{OTOC}^{(k)} :=  \mathrm{OTOC}^{(k)}(O_N,X )  - \mathrm{OTOC}^{(k)}(\tilde{O}_{N,\chi},X )$ for some bounded and normalized $X$. 
    \end{thm}
    \begin{hproof}
    We first identify that the OTOC difference can be written as a sum of $k$ terms, each of which is linear in $\Delta O$, 
    \begin{equation}
        \frac{1}{D}| \sum_{j=0}^{k-1} \tr[ (O_N X)^j (O_N-\tilde{O}_{N,\chi}) X (\tilde{O}_{N,\chi} X)^{k-1-j} ] |.
    \end{equation}
    Then, we apply the triangle inequality, the trace H\"older inequality, and Schatten norm inequalities to isolate the error $\| \Delta O \|_2/\sqrt{D}$. The proof then proceeds as in Thm.~\ref{thm:itac}, with the full details to be found in App.~\ref{ap:otoc}.
\end{hproof}
While we have considered the most common OTOCs over infinite temperature states in Eq.~\eqref{eq:otocDef}, it is immediate to extend Thm.~\ref{thm:otoc} also to low-average ensembles. The upper-bound is almost identical to Thm.~\ref{thm:itac}, up to the coefficient which depends on $\| \tilde{O}_{N,\chi}\|_\infty$. For $k=2$, the coefficient reduces to $2$ and so the conclusions of Eq.~\eqref{eq:itac2} hold: $4$-point OTOCs can be efficiently approximated with an MPO, given $\mc{O}(\log(N))$ scaling of the $\alpha < 1 $ LOE (completing the proof of Table~\ref{tab:results}). We can therefore rigorously confirm the claims of Ref.~\cite{swingle2020}: logarithmic LOE resultant from an operator's light-cone at short times means efficient MPO approximability of OTOC$^{(2)}$. Note also that $\alpha\geq 2$ LOE entropies tightly bound OTOC$^{(2)}(O_N,X)$ when $X$ is chosen randomly~\cite{dowling2023scrambling,Dowling2025thesis}.

For higher-order OTOCs, to ensure an error of $\varepsilon$, it is sufficient that,
\begin{equation}
    \chi \geq (1-\alpha) \left({k N\|\tilde{O}_{N,\chi} \|_\infty^{k-2}}{\varepsilon^{-1}} \right)^{\frac{2\alpha}{1-\alpha}}2^{E^{(\alpha)}(O_N)},
\end{equation}
where we have assumed that $\|\tilde{O}_{N,\chi} \|_\infty \geq 1$. We see that as long as $\|\tilde{O}_{N,\chi} \|_\infty \sim \mathrm{poly}(N)$, then higher-order OTOCs are efficiently approximable using an MPO when $E^{(\alpha)}(O_N) = \mc{O}(\log(N))$, according to Def.~\ref{def:sim}. However, $\| \tilde{O}_{N,\chi} \|_\infty$ can be exponentially large in the worst-case, and is related to Eq.~\eqref{eq:upperHold} via the bound: $\|\tilde{O} \|_\infty \leq \|O \|_\infty + \|\Delta O \|_\infty \leq 1 + \|\Delta O \|_\infty$. We will therefore now study its typical behavior for physical models. 

\begin{figure}[t]
    \centering
    \includegraphics[width=\linewidth]{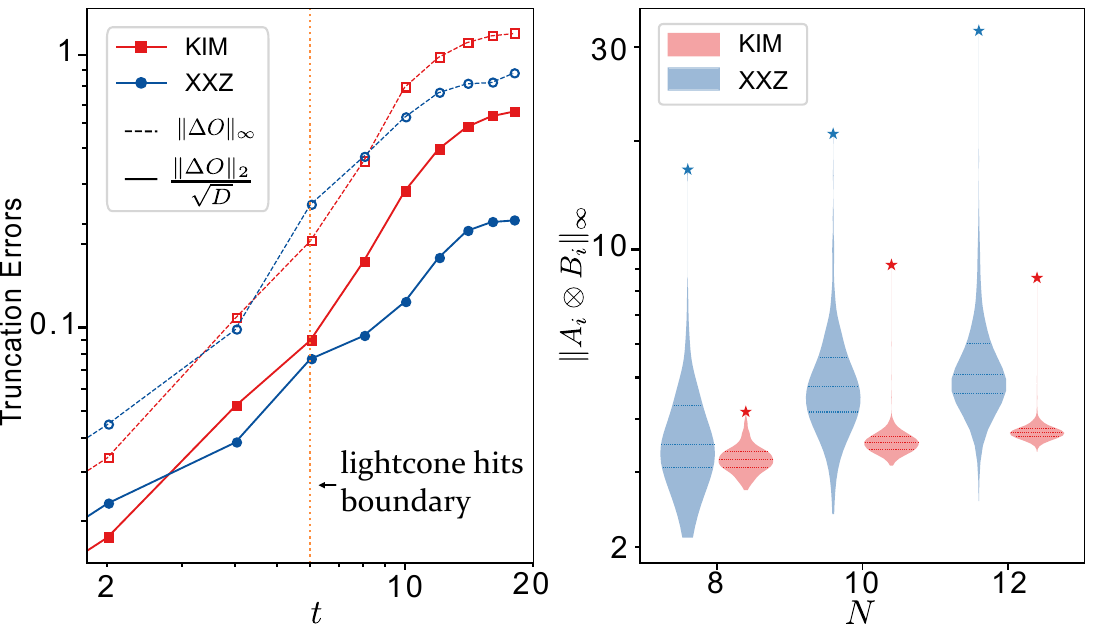}
    \caption{Numerical results on 1D brickwork dynamics of an initially local $\sigma_z$ on the center site, for the XXZ model [blue, circles] and the kicked Ising model (KIM) [red, squares]. On the left, for $N=12$ we plot the spectral [solid] and the (normalized) Hilbert-Schmidt norm [dotted] of $\Delta O = O_N - \tilde{O}_{N,\chi}^{(N/2)}$, given a half-chain truncation with a cutoff of $\varepsilon = 0.02$. On the right, we plot the distribution of spectral norms of the matrices $A_i \otimes B_i$ from the corresponding operator Schmidt decomposition, after $t=20$ layers of evolution. The maximum values are shown as a star, and the quartiles are marked.
    }
    \label{fig:numerics}
\end{figure}

% \textit{Numerical results.---}
\section{Numerical results}
We return our original problem of the distinct truncation errors $\|\Delta O \|_\infty$ versus $D^{-1/2}\|\Delta O \|_2$, bounding non-equilibrium and low-average expectation values, respectively.
We study these quantities numerically in two spin chain models; see Fig.~\ref{fig:numerics}, with further details in App.~\ref{ap:numerics}. We consider an initial $\sigma_z$ operator on the $N/2$ site of an $N$-qubit spin chain, Heisenberg evolving according to a brickwork circuit. In each layer of evolution, $2$-qubit gates act on alternating next-neighbor sites: $(i,i+1),(i+2,i+3),\dots$, then $(i-1,i),(i+1,i+2),\dots$, and so on. For the two-qubit gates, we choose the (interacting integrable) XXZ model
\begin{equation}
    U_{\mathrm{XXZ}}^{(2)}  = \exp\big( -i J (  \sigma_x \otimes \sigma_x + \sigma_y \otimes \sigma_y  + \Delta \sigma_z \otimes \sigma_z) )\big) \label{eq:xxz}
\end{equation}
with $J=1$, $\Delta = 0.55$, and the (non-integrable) kicked Ising model (KIM), 
\begin{equation}
    U_{\mathrm{KIM}}^{(2)} = \exp\big( -i ( J \sigma_z \otimes \sigma_z + \sum_{a=x,z} h_a( \id \ot \sigma_a +  \sigma_a \ot \id)) \big), \label{eq:kim}
\end{equation}
with $J=1,\, h_x=0.9045,\, h_z=0.8090 $. Choosing a half-chain bipartition, in the left panel of Fig.~\ref{fig:numerics} for $N=12$ we plot the two truncation errors for a sharp cutoff of $\lambda_{\chi} \geq 0.02 > \lambda_{\chi+1}$ in the operator Schmidt decomposition, Eq.~\eqref{eq:schmidt}. We find that $\|\Delta O \|_\infty$ is not significantly larger than $D^{-1/2}\|\Delta O \|_2$, despite what is allowable from the worst-case bound. Interestingly, the two truncation errors trend similarly as a power-law with time, with their difference remaining relatively constant.

In the right panel of Fig.~\ref{fig:numerics}, for the same models we compute the distribution of $\| A_i \otimes B_i\|_\infty$ from Eq.~\eqref{eq:schmidt}, after $t=20$ layers with varying $N$. Both $\|\Delta O \|_2$ and the LOE entropies are independent of the structure of $A_i \otimes B_i$, so any varying behavior of $\|\Delta O \|_\infty$ must be encoded therein. 
We see that while the maximum value tends to (almost) saturate the upper bound of $\| A_i \otimes B_i\|_\infty \leq \| A_i \otimes B_i\|_2 = \sqrt{D}$, the distribution is highly concentrated around a value that grows slowly with $N$. We also see a qualitative difference between the integrable and the non-integrable models, with the KIM leading to a more concentrated distribution. It is an interesting open question whether this behavior is a universal characteristic of (non-)integrability.
% In App.~\ref{ap:numerics}, we also plot the distribution for the integrable transverse field Ising model ($h_z=0$ in Eq.~\eqref{eq:kim}), finding values highly concentrated around $\| A_i \otimes B_i\|_\infty=1$.

\section{Random Matrix Model}
% \textit{Random Matrix Model.---}
Motivated by the results of Fig.~\ref{fig:numerics}, we now introduce a random matrix model for the components $A_i \otimes B_i $ appearing in the operator Schmidt decomposition, Eq.~\eqref{eq:schmidt}. Namely, for a given Schmidt spectrum $\{\lambda_i \}$, we replace $A_i$ and $B_i$ with an ensemble of bounded matrices, which, remarkably, allows us to bound the spectral truncation error in terms of LOE entropies.
\begin{thm} \label{thm:random}
    Consider an ensemble of operators such that each element can be written as $O_N = \sum_{i=1}^{d^{2n}} \lambda_i A_i \otimes B_i$ across the bipartition $\mc{H}=\mc{H}_A\otimes \mc{H}_B$ with $N_A =n$. We take the set $\{\lambda_i\}$ to be constant and $A_i \otimes B_i$ sampled independently such that, almost surely: $\mathbb{E} (X)  =0$, $\|X\|_\infty \le L$, and $\mathbb{E} ( X^\dagger  X ) = \mathbb{E} ( X X^\dagger ) = \id$ for $X \in \{A_i,B_i\}$. Then, on average for $0<\alpha <1$,
    \begin{equation}
        \mathbb{E}(\| O_N - \tilde{O}_{N,\chi}^{(n)}\|_\infty)\! \leq \!c N \exp_2 \Bigg( \!\frac{1-\alpha}{2\alpha} \!\Big(\! E^{(\alpha)}_A - \log({\chi}) \!\Big) \!\Bigg) \nn
        % \label{eq:rmt}
    \end{equation}
    where $\tilde{O}_{N,\chi}^{(n)} := \sum_{i=1}^{\chi} \lambda_i A_i \otimes B_i$ and $c$ is a constant depending on $d,\, \alpha,$ and $L$, which is reported in Eq.~\eqref{eq:const}.
\end{thm}
\begin{hproof}
    The proof follows from first applying the matrix Bernstein inequality to $\Delta O$~\cite{tropp2015introductionmatrixconcentration}. Then, the resultant upper bound involves both the Hilbert-Schmidt truncation error $\sum_{i=\chi+1}^{D_A^2} |\lambda_i|^2$ as well as a term proportional to $|\lambda_{\chi+1}|$. The former we can upper-bound using the argument of Thm.~\ref{thm:itac}, while the latter can be bound by LOE through a similar majorization argument. A full proof is supplied in App.~\ref{ap:random}
\end{hproof}
 The rationale of this ensemble is that by only randomizing the matrices $A_i$ and $B_i$, the coefficients $\lambda_i$---and hence the LOE entropies of $O_N$---remain constant. It is therefore meant to represent a typical case of an operator Schmidt decomposition.
 % An example is $A_i \to U^{\dagger} A_i U$ and $B_i \to V^{\dagger} B_i V$ for a Haar random $U,V$ and $\| A_i\|_\infty, \| B_i\|_\infty \leq L$. 
 Then Thm.~\ref{thm:random} together with Eq.~\eqref{eq:upperHold} means that for $E^{(\alpha)}(O_N) = \mc{O}(\log(N))$, on-average $O_N$ can be truncated to a polynomial Schmidt rank, while accurately reproducing expectation values for any state. We stress that this is a toy model; it is an open question to determine relevant dynamics that would lead to such an operator ensemble for every bipartition.

\section{Discussion}
% \textit{Discussion.---} 
We have formalized the intuition that LOE scaling governs the simulability of Heisenberg operators as MPOs. 
An immediate question is, to what degree do our results hold for density matrices? Mixed states with area-law operator entanglement may not be well-approximated by a matrix product density operator (MPDO). This stems from the fact that the trace norm is the natural distance measure between states, bounding any expectation value error (cf. Eq.~\eqref{eq:upperHold}). However, the Hilbert-Schmidt distance is that which can be directly related to R\'enyi entropies~\cite{Verstraete2006,Schuch_2008}, and these two metrics are proportional only for pure states~\cite{wilde_2017}. The entanglement of purification is the more relevant measure when examining representability as MPDOs~\cite{Schuch2020}.

Regarding LOE in many-body physics, questions remain about the extent to which the slow growth of LOE in integrable systems persists in higher dimensions or in other phases. For instance, many-body localized dynamics are thought to have a logarithmic light-cone~\cite{Abanin2019,Zhou2017,kim2014localintegralsmotionlogarithmic}, which should suppress LOE growth~\cite{pineda2026operatordel}. In addition, depolarizing noise dampens high-weight Pauli components of Heisenberg operators~\cite{Rakovszky2022,schuster2024polynomialtime}, and this phenomenon will clearly also inhibit LOE growth~\cite{dowling2024magicheis,Dowlin2024LOE-OSRE}.

One should also compare H-DMRG techniques to other time evolution methods in the Heisenberg picture, such as Pauli propagation~\cite{Rakovszky2022,begusic2024realtime,schuster2024polynomialtime,angrisani2024classically,dowling2024magicheis,rudolph2025paulipropag}. This method involves the truncation of Pauli strings with low amplitude (or weight) in the operator decomposition in this basis, cf. truncating operator Schmidt components. The operator Stabilizer entropy (OSE), a measure of magic resources~\cite{dowling2024magicheis}, plays an equivalent role to the LOE in bounding the cost of Pauli Propagation~\cite{dowling2024magicheis,shao2025pauliprop}, and Thms.~\ref{thm:micro}-\ref{thm:random} can be extended to this setting~\cite{dowlingOSEinprep}. The OSE bounds the LOE~\cite{Dowlin2024LOE-OSRE}, which means that in the scaling limit, H-DMRG should always perform at least as well as Pauli propagation. 
A caveat is that discarding Pauli strings is, computationally, significantly easier to do in practice than the singular value decomposition involved in MPO truncation. A promising direction is then to develop new methods by amalgamating principles behind both Pauli Propagation and H-DMRG; for instance, by truncating in a way that minimizes the spectral norm error.
We leave this to future work.

\twocolumngrid
\begin{acknowledgments}
    AI tools were used to perform a literature search on random matrix theory, which helped lead to Theorem 4. The main ideas, all other proofs, the numerics, and the writing were all done without the use of AI. The author thanks Matteo Rizzi and Norbert Schuch for useful discussions, and Pavel Kos and Gregory White
    for useful discussions and helpful comments on the manuscript. The author acknowledges funding by the Deutsche Forschungsgemeinschaft (DFG, German Research Foundation) under Germany’s Excellence Strategy - Cluster of Excellence Matter and Light for Quantum Computing (ML4Q) EXC 2004/1 - 390534769. 
\end{acknowledgments}

% \makeatletter
% \renewenvironment{thebibliography}[1]
%      {\section*{\refname}%
%       \small
%       \list{\@biblabel{\@arabic\c@enumiv}}%
%            {\settowidth\labelwidth{\@biblabel{#1}}%
%             \leftmargin\labelwidth
%             \advance\leftmargin\labelsep
%             \setlength{\itemsep}{0pt}
%             \setlength{\parsep}{0pt}
%             \usecounter{enumiv}%
%             \let\p@enumiv\@empty
%             \renewcommand\theenumiv{\@arabic\c@enumiv}}%
%       \sloppy\clubpenalty4000\widowpenalty4000%
%       \sfcode`\.\@m}
%      {\def\@noitemerr
%        {\@latex@warning{Empty `thebibliography' environment}}%
%       \endlist}
% \makeatother

% \setlength{\bibitemsep}{0pt}
% \setlength{\bibnamesep}{0pt}
% \setlength{\bibinitsep}{0pt}
% \onecolumngrid
% \bibliographystyle{unsrt}
% \bibliographystyle{quantum}
% \bibliography{references}
% \printbibliography


% \newpage

\onecolumngrid
% \appendix

% \newpage
% \section*{End Matter}
% \twocolumngrid

\onecolumngrid
\section*{Appendix}

\appendix
% \newpage
% switch to a different TOC file
\makeatletter
\newcommand{\appendixtableofcontents}{%
  \subsection*{Contents}
  \@starttoc{atoc}
}
\newcommand{\appendixcontentsline}[3]{%
  \addcontentsline{atoc}{#1}{#2}
}
\makeatother
% \resumetoc
\appendixtableofcontents
% \tableofcontents

\section{Proofs of Main Results} \appendixcontentsline{section}{A. Proofs of Main Results}{} \label{ap:proofs}
Here, we provide full proofs and additional details for the results presented in the main text. We first introduce the setting and notation, and recall some useful intermediate results.

We consider a (family of) operators $O_N$, defined on $N$ sites of a 1D chain with total dimension $D=d^N$,
\begin{equation}
    O_N =  \includegraphics[scale=1, valign=c]{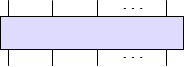},
\end{equation}
where we use standard tensor network graphical notation~\cite{ORUS2014117}, with bra [ket] indices at the bottom [top]; cf. Fig.~\ref{fig:main-results}. Its state representation is then
\begin{equation}
    |O_N\rrangle:= (O_N \otimes \id) \ket{\phi^+} =  \frac{1}{\sqrt{D}}\, \includegraphics[scale=1, valign=c]{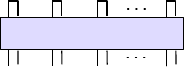}.
\end{equation}
The operator has LOE R\'enyi entropies defined by
\begin{equation}
    E^{(\alpha)}_A(O_N):= (1-\alpha)^{-1} \log( \tr[\tr_{B}[|O_N\rrangle\!\llangle O_N|]^{\alpha}] ) , \label{eq:loe_again}
    % = (1-\alpha)^{-1} \log( \tr[T_{(12)\dots ((\alpha-1) \alpha) } T_{ (2 3) \dots (\alpha 1) } O_N^{\otimes 2 \alpha}|]),
\end{equation}
which we notice are functions of $2\alpha$-copies of $O_N$~\cite{Dowlin2024LOE-OSRE}. These are defined across some bipartition of the Hilbert space $\mc{H}=\mc{H}_A \otimes \mc{H}_B$, where $\mc{H}_A$ include the left $N_A$ qudits and $\mc{H}_B$ the right $N_B=N-N_A$ qudits. Then, recalling that $|O_N\rrangle\in \mc{H}^{\otimes 2}$, the partial trace in Eq.~\eqref{eq:loe_again} is with respect to the doubled-space bipartition, $\mc{H}_A^{\otimes 2} \otimes \mc{H}_B^{\otimes 2}$,
\begin{equation}
    \tr_{B}[|O_N\rrangle\!\llangle O_N|] = \includegraphics[scale=1, valign=c]{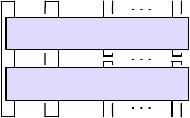}
\end{equation}
In the following proofs, we will consider two approximations to $O_N$. First, if we truncate across a single bipartition between the left $n$ qudits and the right $N-n$, we have that
\begin{equation}
    O_N 
    % \underset{{\text{local basis}}}
    {=}\sum_{i,j=1}^{d^{2 N_A},d^{2 N_B}} c_{ij} {P_i \otimes P_j} \underset{\text{Schmidt}}{=} \sum_{i=1}^{\min(d^{2 N_A},d^{2 N_B})} \lambda_i A_i \otimes B_i \underset{\mathrm{truncate}}{\to} \, \sum_{i=1}^{\chi} \lambda_i A_i \otimes B_i =: \tilde{O}_{N,\chi}^{(n)}. \label{eq:truncateOnce}
\end{equation}
with $A_i$ [$B_i$] having support on the first $N_A=n$ [final $N_B=N-n$] sites, and $P_i$ some local, orthonormal operator basis, as discussed below Eq.~\eqref{eq:mpo}. Without loss of generality, we take $N_A \leq N_B$ throughout, such that $\min(d^{2 N_A},d^{2 N_B}) = d^{2n}$. The first equality in Eq.~\eqref{eq:truncateOnce} is the rewriting of $O_N$ in a local, orthonormal basis. The next comes from the Schmidt decomposition of $| O_N \rrangle$ across the chosen bipartition $A:B$, which can be determined from the singular value decomposition of the correlation matrix $c_{ij}$. Then the set $\{ |\lambda_i|^2 \}_{i=1}^{d^{2n}}$ defines a normalized probability distribution, and its classical R\'enyi entropies are exactly the LOE R\'enyi entropies across the corresponding bipartition. Based on this fact, we recall some useful results on relating the tails of (classical) probability distributions to their entropies~\cite{Verstraete2006,Schuch_2008}.
 \begin{lemma}  \label{lem:majorize}
    Consider the ordered distribution $\{ p_i \}_{i=1}^{r}$ with $p_i \geq p_{i+1}$, R\'enyi entropies $S^{(\alpha)}(\{ p_i\}) := (1-\alpha)^{-1}\log(\sum_i p_i^{\alpha} )$, and tail sum $p := \sum_{i=\chi+1}^r p_i$. Then,
    \begin{enumerate}
        \item [\textbf{a}.] \textup{(\cite{Verstraete2006})} For $0< \alpha < 1$, 
            \begin{equation}
            p \leq \exp_2\left( \frac{1-\alpha}{\alpha}\big( S^{(\alpha)}(\{ p_i\}) - \log\big( \frac{\chi}{1-\alpha} \big) \Big) \right),
            \end{equation}
        \item [\textbf{b}.] \textup{(\cite{Schuch_2008})} For $\alpha > 1$, 
            \begin{equation}
                p \geq 1-\exp_2\left(\frac{\alpha-1}{\alpha }\big(\log(\chi) - S^{(\alpha)}(\{ p_i\}) \big) \right). 
            \end{equation}   
        \item [\textbf{c}.] For $0< \alpha < 1$,
        \begin{equation}
            {p_{\chi +1}} \leq \exp_2 \left( \frac{1-\alpha}{\alpha} S^{(\alpha)}(\{ p_i\}) - \frac{\log(\chi)}{\alpha}\right). \label{eq:a3}
        \end{equation}
    \end{enumerate}
\end{lemma}
\begin{proof}
    The proofs of parts $\mathbf{a.}$ and $\mathbf{b.}$ are provided in Refs.~\cite{Verstraete2006} and \cite{Schuch_2008}, respectively, and rely on majorization arguments that dictate which distribution has the maximal/minimal entropies for a given tail sum. To prove part $\mathbf{c.}$, we follow a similar strategy. We fix $p_{\chi +1 }=x$ and $\chi$, and recognize that the distribution with the minimum entropy satisfying these constraints is: $\{p_1=1-\chi x,\, p_2=p_3=\dots=p_{\chi +1 }=x,\, p_{\chi+2}=\dots=0\}$. This can be understood from the fact that this distribution majorizes all others that satisfy these constraints, and then applying the Schur concavity property of R\'enyi entropies~\cite{bhatia97}. Then, we have that 
    \begin{align}
        S^{(\alpha)}(\{ p_i\}) &\geq \frac{1}{1-\alpha }\log(\chi x^{\alpha} + (1-\chi x)^\alpha ) \\
        &\geq \frac{1}{1-\alpha }\log(\chi x^{\alpha} )\\
        &=\frac{1}{1-\alpha } ( \log(\chi) + {\alpha} \log(x) ).
    \end{align}
    Here we have used that $(1-\alpha)^{-1}$ is positive for $\alpha < 1$ and that the logarithm is an increasing function. Rearranging this expression, and recalling that $x = p_{\chi +1 }$, we arrive at Eq.~\eqref{eq:a3}.    
\end{proof}
Parts \textbf{\textit{a}} and \textbf{\textit{b}} are key to foundational results in MPS simulability~\cite{Verstraete2006,Schuch_2008}, while part \textbf{\textit{c}} is a new result in a similar spirit. 

We will also, of course, consider full MPO approximations to $O_N$, defined in Eq.~\eqref{eq:mpo}, which we reproduce here for convenience,
\begin{equation}
    \begin{split}
        \tilde{O}_{N,\chi} &= \sum_{i_1,\dots,i_N=0}^{d^2-1} \tr[\Lambda^{(1)}_{i_1} \dots \Lambda^{(N)}_{i_N} ] P_{i_1} \otimes \dots \otimes P_{i_N}\\
        &= \includegraphics[scale=1.2, valign=c]{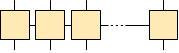} \label{eq:mpo_again}
    \end{split}
\end{equation}
Throughout, we use $\tilde{O}_{N,\chi}^{(n)}$ to indicate an approximation to $O_N$ obtained through a single rank $\chi$ truncation of $O_N$ between the first $n$ qudits and the rest, while $\tilde{O}_{N,\chi}$ is any MPO approximation to $O_N$ which has an operator Schmidt rank $\chi$ across every cut, i.e. for each $n=1,2\dots,N-1$.

% \textit{Appendix A: Proof of Theorem~\ref{thm:micro}.---}
% \tableofcontents
\subsection{Proof of Theorem~\ref{thm:micro}} \label{ap:nogo}
\appendixcontentsline{subsection}{1. Proof of Theorem~1}{}
To prove the no-go result, we need to show the negation of Def.~\ref{def:sim} for an $O_N$ which satisfies either: (i) $E^{(1)}_A(O_N) = \Omega(N)$ or (ii) $E_A^{(\alpha )}(O_N) = \Omega(N^c)$ for $\alpha > 1$ and $c>0$. The beginning of the proof proceeds the same for both cases, so we refrain from assuming either for now. We note that the first half of the proof is novel, while the second half (from below Eq.~\eqref{eq:truncation}) proceeds through a direct application of the results of Ref.~\cite{Schuch_2008}. It suffices to show that, given the LOE scaling assumption, there exists a state $\rho$ and some error $\varepsilon$, such that any MPO $\tilde{O}_{N,\chi}$ that approximates $\tr[O_N \rho]$ up to $\varepsilon$ must have bond dimension $\chi = \omega(\mathrm{poly}_{\varepsilon}(N))$. As an MPO is constructed from successive such approximations, it is enough to consider a truncation across a single bipartition, resulting in $\tilde{O}_{N,\chi}^{(n)}$ from Eq.~\eqref{eq:truncateOnce}.
Further truncations across other cuts may only worsen the error. For the result of Thm.~\ref{thm:micro} to hold, the bipartition need only be one where the LOE has the specified fast scaling. WLOG, from now we take the region $B$ to be larger than $A$ and each marginal system to consist of an extensive number of qudits: $N_B \geq N_A=n =\mc{O}(N)$ (this must be true to have a volume-law LOE scaling). 

We first notice that for Hermitian $O_N$, $\tilde{O}_{N,\chi}^{(n)}$ can also be chosen to be Hermitian, and therefore so is $\Delta O$. This can be understood through choosing a local, Hermitian operator basis for the initial decomposition on the left-hand side of Eq.~\eqref{eq:truncateOnce}. For qubits ($d=2$), one can take $P_i$ and $P_j$ to be Pauli strings on $N_A$ and $N_B$ qubits, respectively (cf. Eq.~\eqref{eq:mpo}). For qudits where $d$ is not a power of $2$, such a basis can always be found (e.g., for qutrits, one could choose the Gell-Mann basis).

Then, we can choose $\rho := |{\psi_{\lambda_{\chi+1}}} \rrangle \llangle{\psi_{\lambda_{\chi+1}}} | $ to be an eigenstate of $\Delta O^{(n)}$ corresponding to the largest eigenvalue, $\lambda_{\chi+1}$. In this case, the bound of Eq.~\eqref{eq:upperHold} becomes an equality: $|\tr[\Delta O^{(n)} \rho ]| = \| \Delta O^{(n)}\|_{\infty}$. It remains to lower-bound $\| \Delta O^{(n)}\|_{\infty}$. From a standard Schatten norm inequality, 
    \begin{equation}
        \| \Delta O^{(n)}\|_{\infty} \geq D^{-1/2} \| \Delta O^{(n)}\|_{2} = \sqrt{D^{-1}\sum_{i,j=\chi+1}^{d^{2n}} \lambda_i \lambda_j^* \tr[(A_i \otimes B_i)^{\dagger} (A_j \otimes B_j)] } = \sqrt{\sum_{i=\chi+1}^{d^{2n}} |\lambda_i|^2},\label{eq:trunc_error}
    \end{equation}
    where the equality on the right-hand side can be derived from the Hilbert-Schmidt orthonormalization of $A_i \otimes B_i$. We define the two (normalized) pure states 
    \begin{equation}
        \begin{split}
            &| O_N \rrangle := (O_N \otimes \mathbb{1})|\phi^+\rangle,  \quad \text{and}  \\
        &| \tilde{O}_{N,\chi}^{(n)} \rrangle := \frac{1}{\sqrt{\sum_{i=1}^\chi \lambda_i^2}}(\tilde{O}_{N,\chi}^{(n)} \otimes \mathbb{1})|\phi^+\rangle, \label{eq:CJI}
        \end{split}
    \end{equation}
    where the first is the familiar Choi state of $O_N$ as described above Eq.~\eqref{eq:mpo}, while the second is the normalized Choi state of $\tilde{O}_{N,\chi}^{(n)}$. Then, the quantity on the right-hand-side of Eq.~\eqref{eq:trunc_error} is exactly the trace distance between these two states, 
    \begin{align}
    \frac{1}{2}\| |{O_N}\rrangle \llangle O_N |  - |{\tilde{O}_{N,\chi}^{(n)}}\rrangle \llangle {\tilde{O}_{N,\chi}^{(n)}} | \|_1 &= \sqrt{1 - |\llangle O_N|\tilde{O}_{N,\chi}^{(n)}\rrangle|^2}, \\
    &=  \sqrt{1 - \frac{1}{D^2} \left(  \frac{1}{\sqrt{\sum_{i=1}^\chi \lambda_i^2}} \tr[{O}_t^{\dagger } \tilde{O}_{N,\chi}^{(n)} ]\right)^2 } \\
    &= \sqrt{1 - \frac{1}{D^2} \left( \frac{\sum_{i=1}^\chi \lambda_i^2}{\sqrt{\sum_{i=1}^\chi \lambda_i^2}}\tr[\id]\right)^2 }\\
    &=\sqrt{1 - {\sum_{i=1}^\chi \lambda_i^2} } = \sqrt{\sum_{i=\chi+1}^{D_A^2} \lambda_i^2}, \label{eq:truncation}
\end{align}
where we have used that $\sum_{i =1}^{D_A^2} \lambda_i^2 = 1$.
% , which comes from the normalization $\| O \|_2 = \sqrt{D}$ and hence the normalization of the Choi state, $|{O_N}\rrangle = D^{-1/2} (O_N \otimes \id)\ket{\phi^+}$.
    % The right-hand-side is exactly the MPS truncation error for the quantum state $|O \rrangle$, appearing, for example, in the proofs of Refs.~\cite{Verstraete2006,Schuch_2008,Schuch2020}.
    Recalling that the squared singular values $\{ |\lambda_i|^2 \}_{i=1}^{D_A^2}$ define the probability distribution from which the LOE entropies are defined, the remainder of the proof proceeds via applying arguments from the two main results of Ref.~\cite{Schuch_2008}: 
    \begin{itemize}
        \item [(i)] We first provide a bound in terms of $E_A^{(1)}(O_N)$. Define the reduced density matrices of the states in Eq.~\eqref{eq:CJI}: $\sigma:= \tr_{B}[| O_N \rrangle \! \llangle  O_N |]$ and $\tilde{\sigma}:= \tr_{B}[| \tilde{O}_{N,\chi}^{(n)}\rrangle \! \llangle \tilde{O}_{N,\chi}^{(n)} |]$. Then, from Eq.~\eqref{eq:loe}, the R\'enyi entropy of $\sigma$ is exactly the LOE R\'enyi entropy, while the entropy of $\tilde{\sigma}$ satisfies $S^{(1)}_A(\tilde{\sigma}) \leq S^{(1)}_A(\sigma)$ and $S^{(1)}_A(\tilde{\sigma}) \leq \log(\chi)$.  From the contractivity of the trace norm under partial trace, we know that  
        \begin{equation}
            \| \sigma - \tilde{\sigma}\|_1 \leq \|| O_N \rrangle \! \llangle  O_N | - | \tilde{O}_{N,\chi}^{(j)} \rrangle \! \llangle \tilde{O}_{N,\chi}^{(j)}| \|_1.
        \end{equation}
        Applying the above, together with the Fannes-Audenaert inequality~\cite{wilde_2017} and Eqs.~\eqref{eq:trunc_error} and \eqref{eq:truncation}, we find, 
        \begin{align}
            S^{(1)}_A(\sigma)- \log(\chi)  &\leq |S^{(1)}_A(\sigma)- S^{(1)}_A(\tilde{\sigma})| \\
            \iff E^{(1)}_A(O_N)- \log(\chi)&\leq \frac{1}{2}\|\sigma-\tilde{\sigma} \|_1 \log(D_A^2) +1 \\
            &\leq  \| \Delta O^{(n)}\|_{\infty} 2 N_A  \log(d) +1.
        \end{align}
        Choosing $\rho$ such that $\| \Delta O^{(n)}\|_{\infty} = |\tr[\Delta O \rho]| =:\varepsilon $ (as discussed above Eq.\eqref{eq:trunc_error}), we arrive at 
        \begin{equation}
            \log(\chi) \geq  E^{(1)}_A(O_N) -2\varepsilon N_A\log(d)-1. \label{eq:a23}
        \end{equation}
        Now, we will interpret this in terms of approximability as an MPO, Def.~\ref{def:sim}.  Assume that $E_A^{(1)}(O_N) = \Omega(N)$. We know that the size of the bipartition $A$ must scale with $N$ to possibly have a volume law LOE R\'enyi entropy, so $N_A \geq \eta N $ for some $\eta >0$, and by our assumption we also know that $E^{(1)}_A(O_N) \geq a N $ for some $a>0$. Choosing a sufficiently small: $\varepsilon \leq a/(2 \eta \log(d)) = \mc{O}(1)$, we find that $\chi = \Omega(\exp(N)) $. Therefore, when $E^{(1)}_A(O_N) = \Omega(N)$, we have shown that there is no polynomial bond dimension MPO that well-approximates $O_N$ for expectation values over all possible states. Finally, we pick the worst case, $N_A \to N/2 $, to arrive at the first bound of Thm.~\eqref{thm:micro}.
        
        \item [(ii)] Now we consider $E_A^{(\alpha)}(O_N)$ for $\alpha > 1$. Again, we proceed by lower-bounding the Hilbert-Schmidt truncation error ${\sum_{i=\chi+1}^{D_A^2} |\lambda_i|^2}$, which lower-bounds $\|\Delta O^{(n)} \|_\infty =: \varepsilon$, which in turn is equal to the expectation value error for $\rho$ chosen as described above Eq.~\eqref{eq:trunc_error}. From Lemma~\ref{lem:majorize}.b, for $\alpha> 1$ by identifying $\{ p_i\}_{i=1}^r \to  \{|\lambda_i|^2 \}_{i=1}^{d^{2n}}$ and $S^{(\alpha)}(\{ p_i\} ) \to E^{(\alpha)}_A(O_N)$, we have that 
        \begin{equation}
                \varepsilon^2 \geq \sum_{i=\chi+1}^{D_A^2} |\lambda_i|^2 \geq 1-\exp_2\left(\frac{\alpha-1}{\alpha }\big(\log(\chi) - E^{(\alpha)}_A(O_N) \big) \right). \label{eq:a24}
            \end{equation}
        Rearranging the above relation, we again find a lower-bound to the required bond dimension of $\tilde{O}_{N,\chi}$ for the error $\varepsilon$,
        \begin{equation}
            \log(\chi) \geq E^{(\alpha)}_A(O) + \frac{\alpha}{\alpha - 1}\log(1-\varepsilon^2 ).
        \end{equation}
         Assume that $E^{(\alpha)}(O_N) = \Omega(N^c)$ for $\alpha > 1$ and $c>0$. It is immediate to see that for any $a>0$ such that $E^{(\alpha)}(O_N) \geq a N^c $, the bond dimension scales at least as fast as $\chi \sim  \exp( a N^c)$ irrespective of the error $\varepsilon$, which is faster than any polynomial of $N$. Therefore, when $E^{(\alpha)}_A(O_N) = \Omega(N^c) $ with $\alpha >1$ and $c>0$, we have shown that no polynomial bond dimension MPO can well approximate $O_N$ for expectation values over all possible states. 
    \end{itemize}
    Finally, before assuming any scaling of the LOE, we note that the above expressions, Eqs.~\eqref{eq:a23}-\eqref{eq:a24}, are valid for any bipartition $A:B$. We therefore optimize over all bipartitions, 
    \begin{equation}
        E^{(\alpha)}(O_N):= \max_{\mc{H}=\mc{H}_A \otimes \mc{H}_B} \{ E^{(\alpha)}(O_N) \} 
    \end{equation}
    to arrive at the strongest bound of Thm.~\ref{thm:micro}.

\subsection{Proof of Theorem~\ref{thm:itac}} \label{ap:thm2}
\appendixcontentsline{subsection}{2. Proof of Theorem 2}{}
We now consider two-point correlation functions, $\tr[O_N X \rho]$, where we evolve (and truncate) only the operator $O_N$, and otherwise assume that the operator $X$ is Hilbert-Schmidt normalized with bounded operator norm: $\| X\|_2 = \sqrt{D}$ and $\| X \|_{\infty }\leq 1$. For instance, if $X=\id$, then we return to the case of single-point expectation values which are the subject of Thm.~\ref{thm:micro}. We consider a family of ensembles of states defined on $N$ qudits, which we call $\mc{E}$, and whose first moment has a spectral norm satisfying,
\begin{equation}
    \| \int_{\rho\sim \mc{E}}{\rho} \|_\infty := \| \overline{\rho} \|_\infty \leq \frac{b}{D}, \label{eq:spec_norm}
\end{equation}
for $b\geq 1$. Examples of such states include the trivial distribution of a single (high-temperature) state, (high-temperature) Scrooge ensembles, uniform distributions over computational basis states, Haar unitary or Clifford ensembles, and unitary designs. We note that similar `low-average' ensembles are studied in Ref.~\cite{schuster2024polynomialtime}, but here we allow $b$ to possibly scale with $N$. We consider a rank $\chi$ MPO approximation to $O_N$ and write the difference as $\Delta O := O_N - \tilde{O}_{N,\chi}$.

The first step is to show that the Hilbert-Schmidt norm appearing in Eq.~\eqref{eq:itac} also bounds the average mean-square error between expectation values:
\begin{equation}
    \begin{split}
        \int_{\rho \sim \mc{E}} \tr[ \Delta O X \rho ]^2 &= \int_{\rho \sim \mc{E}} \tr[ \sqrt{\rho} \Delta O X \sqrt{\rho} ]^2\\
        % \int_{\rho \sim \mc{E}} \sum_{i,j} p_i p_j    \tr\big[(\ket{\psi_i}\!\bra{\psi_i}) ( \Delta O X \ket{\psi_j}\!\bra{\psi_j} (\Delta O X)^{\dagger} )  \big]  \\
    &\leq  \int_{\rho \sim \mc{E}} \tr[\rho] \tr[\sqrt{\rho} (\Delta O X)^\dagger  (\Delta O X) \sqrt{\rho} ]
    % \int_{\rho \sim \mc{E}} \sum_{i,j} p_i p_j     \bra{\psi_j} \Delta O X X^{\dagger} \Delta O^{\dagger}  \ket{\psi_j} 
    \label{eq:av_proof}\\
    &=     \tr[ (\Delta O X)^2 \overline{\rho}]  \\
    &\leq \| \overline{\rho} \|_\infty  \|(\Delta O X)^2\|_1  \\
    &\leq \frac{b}{D}   \|\Delta O X \|_2^2  \\
    &\leq b \frac{\|\Delta O\|_2^2}{D}.  
    \end{split}
\end{equation}
Here, we have first
% first expanded the arbitrary state $\rho$ within the averaging in terms of a convex sum of pure states, $\rho=\sum_i p_i \ket{\psi_i}\!\bra{\psi_i}$, then 
used the Cauchy-Schwarz inequality with respect to the Hilbert-Schmidt inner product, and finally applied the trace H\"older's inequality. Note also that $\|(\Delta O X)^2\|_1 = \tr[(\Delta O X)^2] = \|\Delta O X\|_2^2 \leq \|\Delta O \|_2^2$, with the final step using the assumption $\| X \|_{\infty} \leq 1$ together with the inequality $\|XY\|_2\leq \| X\|_\infty \|Y\|_2$~\footnote{This can be proven through applying the von Neumann trace inequality to the definition $\|XY\|_2^2 = \tr[Y^\dagger X^\dagger X Y ]$, for ordered singular values $\sigma_i$: $ \|XY\|_2^2 \leq \sum_i \sigma_i(X)^2 \sigma_i(Y)^2 \leq \sigma_1(X)^2 \sum_i \sigma_i(B)^2$.
}. A similar relation between the Hilbert-Schmidt norm and the average mean-square difference of expectation values has been derived in Ref.~\cite{schuster2024polynomialtime}. 

It remains to upper bound the Hilbert-Schmidt norm difference appearing in the upper bound of Eq.~\eqref{eq:av_proof} in terms of LOE R\'enyi entropies. It turns out to be a simpler task first to consider an approximation to $O_N$ found through truncating across a \textit{single} bipartition $A:B$, as in Eq.~\eqref{eq:truncateOnce}, 
\begin{equation}
    \Delta O^{(n)} := O_N - \tilde{O}_{N,\chi}^{(n)} .
\end{equation}
Explicitly, a single-cut truncation is given by
\begin{align}
     \tilde{O}_N^{(n)} := \sum_{i=1}^{\chi} \lambda_i^{(n)} A_i^{(n)} \otimes B_i^{(n)} = P_n(O_N).
\end{align}
where for $1\leq n \leq N-1$, $A_i^{(n)}$ has support on the first $n$ qudits, while $B_i^{(n)}$ has support on the remaining $N-n$ qudits. $P_n$ is an orthogonal projection operation, which can be taken to act either on the first $n$ or final $N-n$ qudits. We know the existence of $P_n$ from the orthonormality of the Schmidt vectors $|A_i^{(n)} \otimes B_i^{(n)} \rrangle$, and hence Hilbert-Schmidt orthonormality of the matrices: $\tr[(A_i^{(n)} \otimes B_i^{(n)}) (A_j^{(n)} \otimes B_j^{(n)})] = D \delta_{ij}$. Explicitly, assuming that $n=N_A < N_B$, we can write
\begin{equation}
        P_n(X) = \frac{1}{d^{N_A}} \sum_{i=1}^{\chi} A_i \otimes \tr_A[(A_i^\dagger \otimes \id_B) X].
\end{equation}
Now we will `stitch together' such single-cut approximations. A full rank-$\chi$ MPO approximation to $O_N$ can be found through successive such projections across every cut~\cite{Schuch2020,jiriThesis},
\begin{equation}
    \tilde{O}_{N,\chi} = P_1 (P_2 (\cdots (P_{N-1} (O_N))\cdots ) ). \label{eq:projections}
\end{equation}
Each projection ensures that the bond dimension across its corresponding cut is limited to $\chi$, so we are guaranteed that the above procedure gives an MPO approximation, but of course, there could be more optimal or efficient truncation methods in practice. Note that the ordering of the projections matters, as they do not commute generally. In terms of standard tensor network diagrammatics (cf. Fig.~\ref{fig:main-results}),
\begin{equation}
    |P_{N-1}(O_N)\rrangle = \frac{1}{\sqrt{D}}\, \includegraphics[scale=1.6, valign=c]{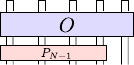} = \frac{1}{\sqrt{D}} \, \includegraphics[scale=1.6, valign=c]{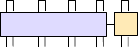} = |\sum_{i_{N-1}=1}^{\chi}  A_{i_{N-1}}^{[1\dots N-1]} \otimes B_{i_{N-1}}^{[N]} \rrangle, \label{eq:series_projs}
\end{equation}
where we work in the vectorized representation, as this is where the projectors act via matrix multiplication. Note also that the singular values are absorbed into the definition of $B_{i_{N-1}}^{[N]}$. 
After $N-1$ projections, given that the successive ones do not cross any previously created bond, we find 
\begin{equation}
        |P_1 (\cdots (P_{N-1} (O_N))\cdots )\rrangle = \frac{1}{\sqrt{D}} \, \includegraphics[scale=1.6, valign=c]{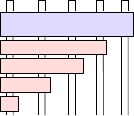} = \frac{1}{\sqrt{D}} \, \includegraphics[scale=1.6, valign=c]{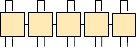} 
\end{equation}
where the right-hand side is exactly a rank $\chi$ MPO, $\tilde{O}_{N,\chi} $. We do not write out the resultant MPO explicitly here and instead refer to tensor network literature where this is a standard approach~\cite{Verstraete2006,Eisert_Cramer_Plenio_2010,ORUS2014117,Schuch2020,jiriThesis}. We now need to relate the successive single-cut errors to the full MPO error $\| \Delta O\|_2 $ appearing in Eq.~\eqref{eq:av_proof}. The projections do not commute in Eq.~\eqref{eq:projections}, so we cannot combine them into a single orthogonal projection operation. However, any two projections give a combined error of~\cite{Schuch2020,jiriThesis}
\begin{align}
    \|O_N - P_m(P_n(O_N ) )\|_2 &\leq \|O_N - P_m(O_N ) \|_2+ \|P_m(O_N) - P_m(P_n(O_N ) )\|_2 \\
    &= \|O_N - P_m(O_N ) \|_2+ \|P_m(O_N - P_n(O_N ) )\|_2 \\
    &\leq \|O_N - P_m(O_N ) \|_2+ \|O_N - P_n(O_N ) \|_2, 
\end{align}
where we have used the triangle inequality and the fact that the Hilbert-Schmidt norm is contractive under orthogonal projections. Applying this inequality $N-1$ times, we can directly bound 
\begin{equation}
    \|\Delta O \|_2 = \|O_N- P_1 (\cdots (P_{N-1} (O_N))\cdots  )\|_2 \leq \sum_{n=1}^{N-1} \|O_N -  P_n(O_N) \|_2 \leq (N-1) \max_{1\leq n \leq N-1} \| \Delta O^{(n)}\|_2. \label{eq:c9}
\end{equation}
For the single-cut approximations appearing on the right-hand side here, $\| \Delta O^{(n)}\|_2$, we recall the relation in Eq.~\eqref{eq:trunc_error} between the sum of truncated singular values and the Hilbert-Schmidt distance: $D^{-1}\| O_N - \tilde{O}_{N,\chi}^{(n)} \|_2^2  = { \sum_{i=\chi+1}^{d^{2 n }} \lambda_i^2 }$. We can now directly apply Lemma~\ref{lem:majorize} to find an upper bound in terms of LOE entropies. Namely, setting $\{ p_i\}_{i=1}^r \to  \{|\lambda_i|^2 \}_{i=1}^{d^{2n}}$ and $S^{(\alpha)}(\{ p_i\} ) \to E^{(\alpha)}_A(O_N)$, we have for $\alpha < 1$, 
\begin{equation}
    \frac{\|\Delta O^{(n)}\|_2^2}{D}\leq    \exp_2\left( \frac{1-\alpha}{\alpha}\Big( E^{(\alpha)}_A(O_N) - \log\big( \frac{\chi}{1-\alpha} \big) \Big) \right), \label{eq:first}
\end{equation}
valid for any bipartition $A:B$, with sizes $N_A = n$ and $N_B = N-n$. Finally, from Jensen's inequality, we know that the root of the average mean-square difference is always an upper bound to the average absolute error. Applying this to Eq.~\eqref{eq:av_proof} together with Eqs.~\eqref{eq:c9} and \eqref{eq:first}, we find
\begin{align}
     \int_{\rho \sim \mc{E}} | \tr[ \Delta O \rho ] | &\leq \sqrt{ \int_{\rho \sim \mc{E}} \tr[ \Delta O \rho ]^2 } \\
     &\leq b^{1/2} \frac{\|\Delta O \|_2}{D^{1/2}} \\
     &\leq (N-1)b^{1/2} \max_{1\leq n \leq N-1} \frac{\| \Delta O^{(n)}\|_2}{D^{1/2}} \\
     &\leq (N-1) b^{1/2} \exp_2\left( \frac{1-\alpha}{2\alpha}\Big( E^{(\alpha)}(O_N) - \log\big( \frac{\chi}{1-\alpha} \big) \Big) \right),
\end{align}
where we take $E^{(\alpha)}(O_N)$ to be the max LOE over all cuts. We have arrived at Thm.~\ref{thm:itac} (where we also write $N-1 \leq N $ for brevity). Analyzing this bound in the context of simulability, if we assume that $E^{(\alpha)}_A(O_N) \leq c \log(N)$, then
\begin{equation}
    \int_{\rho \sim \mc{E}} |\tr[\Delta O X \rho ]| \leq (N-1) b^{1/2}  \left( \frac{N^c (1-\alpha)}{\chi}\right)^{\frac{1-\alpha}{2\alpha} } \label{eq:one}
\end{equation}
and so $|\tr[\Delta O X \rho ]| \leq \varepsilon$ is satisfied on-average when
\begin{equation}
    \chi \geq N^c(1-\alpha)\left( \frac{b  (N-1)^2 }{\varepsilon^2} \right)^{\frac{\alpha}{1-\alpha}}.\label{eq:two}
\end{equation}
We see that if $b\sim \mathrm{poly}(N)$ and the required precision is $\varepsilon^{-1} \sim \mathrm{poly}(N)$, then the required bond dimension $\chi$ is also a polynomial of $N$. From Def.~\ref{def:sim} this means that for the scaling $E^{(\alpha)}_A(O_N) = \mc{O}(\log(N))$ for $\alpha <1$, $O_N$ is efficiently approximable as an MPO on average according to two-point correlations over ensembles of states satisfying Eq.~\eqref{eq:spec_norm} with $b=\mc{O}(\mathrm{poly}(N))$.

\subsection{Proof of Theorem~\ref{thm:otoc}} \label{ap:otoc}
\appendixcontentsline{subsection}{3. Proof of Theorem 3}{}
We now move on to the question of the efficient simulability of higher-order OTOCs using MPOs. Our first task is to relate the difference between OTOCs to a Hilbert-Schmidt distance between the operator $O_N$ and an approximation $\tilde{O}_{N,\chi}$. To reduce notational clutter in this section, we write $O \equiv O_N$ and $\tilde{O} \equiv \tilde{O}_{N,\chi}$. Using a telescoping sum identity, we know that for any matrices $X$ and $Y$,
\begin{equation}
    X^{\otimes k} -Y^{\otimes k} = \sum_{j=0}^{k-1} X^{\ot j} \ot (X-Y) \ot Y^{\ot (k-1-j)}. \label{eq:telescope}
\end{equation}
It is readily checked that the $2k$-OTOC can be rewritten in terms of a trace over $k$ replicas of Hilbert space,
\begin{equation}
    \mathrm{OTOC}^{(k)}(O,X):=\frac{1}{D}\tr[( O X)^{k}] =  \frac{1}{D}\tr[ O^{\otimes k} X^{\otimes k} T_{\gamma}] \label{eq:replica_otoc}
\end{equation}
where $T_{\gamma}$ is the unitary which cyclically permutes the replicas: sending the first to the second, the second to the third, etc. Combining Eqs.~\eqref{eq:telescope} and \eqref{eq:replica_otoc}, and returning back to a single replica representation of the $2k$-OTOC, we can bound the error 
 \begin{align}
    |\Delta \mathrm{OTOC}^{(k)} | &:=|\mathrm{OTOC}^{(k)}(O,X) - \mathrm{OTOC}^{(k)}(\tilde{O},X) |\nn  \\
    &= \frac{1}{D}| \tr[ (O^{\otimes k}-\tilde{O}^{\otimes k}) X^{\otimes k} T_{\gamma}] |\nn  \\
    &=\frac{1}{D}| \sum_{j=0}^{k-1} \tr[ (O^{\ot j} \ot (O-\tilde{O}) \ot \tilde{O}^{\ot (k-1-j)}) X^{\otimes k} T_{\gamma}] | \nn \\
    &=\frac{1}{D}| \sum_{j=0}^{k-1} \tr[ (O X)^j (O-\tilde{O}) X (\tilde{O} X)^{k-1-j} ] |
\end{align}
Applying the triangle and Cauchy-Schwarz inequalities, for $k\geq 2 $ we have
 \begin{equation}
     \begin{split}
         |\Delta \mathrm{OTOC}^{(k)} |  \leq& \frac{1}{D} \sum_{j=0}^{k-1}   \| \Delta O \|_2 \| X (\tilde{O}X)^{k-1-j} (O X)^j \|_2  \\
        \leq& \frac{\| \Delta O \|_2 }{D} (\sum_{j=0}^{k-2} \| O \|_\infty^j \| X \|_\infty^k  \|\tilde{O} \|_\infty^{k-2-j} \| \tilde{O} \|_2 + \| O \|_\infty^{k-1} \| X\|_\infty^k \|\id \|_2)  \\
        \leq&  \frac{\| \Delta O \|_2 }{D} \| (\| \tilde{O} \|_2 \sum_{j'=0}^{k-2}  \|\tilde{O}\|_\infty^{j'} + \|\id \|_2) \\
        \leq &  \frac{\| \Delta O \|_2 }{\sqrt{D}}  \left(  \frac{\|\tilde{O}\|_\infty^{k-1}-1}{\|\tilde{O}\|_\infty-1} + 1 \right)
     \end{split}
 \end{equation}
Above, we have applied the inequality $\|XY\|_2\leq \| X\|_\infty \|Y\|_2$ (which is also used in Eq.~\eqref{eq:av_proof}) and used the assumptions that $\|X\|_\infty,\, \|O\|_\infty \leq 1$ and $\| \tilde{O}\|_2 \leq \| O \|_2 = \sqrt{D}$. Looking at the above relation, we have an overall factor of the familiar (normalized) Hilbert-Schmidt error, ${\| \Delta O \|_2 }/{\sqrt{D}} $, which we can directly bound in terms of the $\alpha <1$ LOE R\'enyi entropies as in Thm.~\ref{thm:itac}. For $k=2$, corresponding to the most-commonly studied $4$-point OTOC, we find a factor of $2$, and so in that case we are done: all the implications of Thm.~\ref{thm:itac} apply. However, for $k\geq 3$ we have a multiplicative factor in terms of $\|\tilde{O}\|_\infty$. In full generality, this quantity could be large, as each of the orthogonal projections in Eq.~\eqref{eq:series_projs} can be naively bounded only by a factor of $\sqrt{\chi}$ without further assumptions.

Finally, we bound the factor ${\| \Delta O \|_2 }/{\sqrt{D}}$ in the same way as Thm.~\ref{thm:itac}: we can bound each single-cut truncation error by $\alpha<1$ R\'enyi entropies using Lemma~\ref{lem:majorize}, and then we stitch together such approximations using the argument of Eq.~\eqref{eq:c9} to arrive at the final inequality of Thm.~\ref{thm:otoc}.

\section{Typical Spectral Truncation Error} \label{ap:random}
\appendixcontentsline{section}{B. Typical Spectral Truncation Error}{}
We now return to the problem of the spectral norm error of MPO approximations. While $\|\Delta O \|_{\infty}$ upper bounds arbitrary non-equilibrium expectation values (see the discussion around Eq.~\eqref{eq:upperHold}), it is not easy to constrain it in terms of LOE entropies. This is because $\|\Delta O \|_{\infty}$ depends strongly on the spectral properties of the singular matrices $A_i \otimes B_i$ from the operator Schmidt-decomposition, while LOE entropies depend only on the distribution of (squared) Schmidt coefficients. One might therefore hope to relate $\|\Delta O \|_{\infty}$ to $\|\Delta O \|_{2}$ (which in turn bounds LOE), in order to obtain simulability guarantees for expectation values over arbitrary out-of-equilibrium states. However, the general bound of $\| \Delta O \|_\infty \leq \| \Delta O \|_2$ can lead to $\| \Delta O \|_\infty$ being exponentially larger than $\| \Delta O \|_2/\sqrt{D}$. Nonetheless, plotting the two errors in relevant spin chain models shows that while $\| \Delta O \|_\infty$ is of course larger than $\| \Delta O \|_2/\sqrt{D}$, they behave somewhat similar and $\| \Delta O \|_\infty$ does not diverge with $N$; see Fig.~\ref{fig:numerics}.

To explain this behavior, in this section, we develop a random matrix model that describes a generic set of matrices appearing in the operator Schmidt decomposition, for a given distribution of Schmidt coefficients $\lambda_i$. In particular, for a given operator Schmidt decomposition, we replace the resultant matrices with a random matrix ansatz:
\begin{equation}
    O = \sum_{i=1}^{d^{2 N_A}} \lambda_i A_i \otimes B_i  \to \mc{E}_O := \left\{O \sim \sum_{i=1}^{d^{2 N_A}} \lambda_i A_i \otimes B_i : \mathbb{E} (A_i \otimes B_i) =0, \mathbb{E} ( A_i A_i^{\dagger}) = \id =\mathbb{E} ( B_i B_i^{\dagger}), \|A_i \otimes B_i\|_\infty \le L \right\}.
\end{equation}
This leaves both the LOE R\'enyi entropies and the Hilbert-Schmidt norm error invariant, as we leave the operator Schmidt spectrum $\{\lambda_i\} $ constant across the ensemble. The first two conditions on $A_i \otimes B_i$ mean that the distribution is centered around the zero matrix, with a square that is centered around $\id$: this is reasonable, as it is an unbiased way to ensure the fact that $\tr[A_i] =\tr[B_i]=0$ and $\tr[A_i^\dagger A_i ] =D_A$ and $\tr[B_i^\dagger B_i ] = D_B$ (which we know are true for any of the traceless, normalized $O$ which one usually considers). The final condition of this ensemble is somewhat more difficult to justify. From numerical results (see Figs.~\ref{fig:numerics}-\ref{fig:numerics2}), we have determined that the spectral norm can approximately reach its upper bound of $\|A_i \otimes B_i \|_\infty \leq \|A_i \otimes B_i \|_2 = \sqrt{D} $. However, we also observe that the bulk of the distribution of $\|A_i \otimes B_i \|_\infty$ are $\mc{O}(1)$, and the distribution becomes more concentrated with increasing $N$.

To arrive at our result, we need a fundamental result from random matrix theory, which we reproduce from Ref.~\cite{tropp2015introductionmatrixconcentration}.
\begin{lemma} \textbf{(Matrix Bernstein Inequality)} \label{lem:bern}
    Consider a finite sequence $\{S_i\} $ of independent, $D$-dimensional random matrices, which satisfy:
\[
\mathbb{E} (S_i) = 0
\quad\text{and}\quad
\|S_i\|_\infty \le L
\quad\text{almost surely for each } i .
\]
Then
\[
\mathbb{E}(\,\|\sum_i S_i\|_\infty)
\;\le\;
\sqrt{\,2\,\nu\,\ln(2D)\,}
\;+\;
\frac{1}{3} L\, \ln(2D),
\]
where $\nu$ is the variance, 
\begin{equation}
\nu := \max \bigl\{\, \bigl\| \mathbb{E}( \sum_i S_i S_i^\dagger) \bigr\|_\infty,\;
                   \bigl\| \mathbb{E}(\sum_i S_i^\dagger S_i) \bigr\|_\infty \bigr\}.
\end{equation}
\end{lemma}
Using this, we can prove Thm~\ref{thm:random}. Namely, considering the operator difference after truncation for a sampling of $\mc{E}_O$, $\Delta O = \sum_{i=\chi+1}^{d^{2N_A}} \lambda_i A_i \otimes B_i$, we choose $S_i :=  \lambda_i A_i \otimes B_i $ with $\chi+1 \leq i \leq d^{2 N_A}$ in the matrix Bernstein inequality. Then, we have 
\begin{align}
    &\mathbb{E} (S_i) \to  \lambda_i \mathbb{E} (A_i \otimes B_i) =0, \\
    &\|S_i\|_\infty \to \| \lambda_i  | \|A_i \otimes B_i \|_\infty \leq | \lambda_{\chi+1}  | L, \\ 
    &\| \mathbb{E}(\sum_{i}  S_i^\dagger S_i)\|_\infty \to \| \sum_{i=\chi+1}^{d^{2N_A}} |\lambda_i|^2 \mathbb{E}(A_i \otimes B_i)  \|_\infty =  \| \sum_{i=\chi+1}^{d^{2N_A}} |\lambda_i|^2 \id \|_\infty = \sum_{i=\chi+1}^{d^{2N_A}} |\lambda_i|^2.
\end{align}
In the second line we have used that $\lambda_i \geq \lambda_{i+1}$. We notice in the final line, we have arrived at our usual Hilbert-Schmidt operator truncation error, $D^{-1}\| \Delta O \|_2^2$. Putting this all together into Lemma~\ref{lem:bern}, we have that 
\begin{align}
    \mathbb{E}(\,\|\sum_i S_i\|_\infty) \to   \mathbb{E}(\,\|\Delta O\|_\infty) &\leq \sqrt{2 D^{-1}\| \Delta O \|_2^2 \ln(2 d^N ) } + \frac{\ln(2 d^N ) | \lambda_{\chi+1}  | L }{3} \\
    &\leq  \frac{\| \Delta O \|_2}{\sqrt{D}}\sqrt{2 (N+1) \ln(d) } + \frac{(N+1)\ln(d)\,  | \lambda_{\chi+1}  | L }{3}  
\end{align}
where in the inequality we use that $d\geq 2 $ to get a slightly cleaner expression. Finally, noticing that $\{|\lambda_i|^2\} $ defines a distribution, we apply Lemma~\ref{lem:majorize} parts \textbf{a} and \textbf{c} to $D^{-1/2}\| \Delta O \|_2$ and $ | \lambda_{\chi+1}  |$ respectively, to arrive at:
\begin{equation}
    \begin{split}
        \mathbb{E}(\,\|\Delta O\|_\infty) \leq&   \sqrt{2 (N+1) \ln(d) } \exp_2 \left( \frac{1-\alpha}{2\alpha} \Big( E^{(\alpha)}(O_N) - \log(\frac{\chi}{1-\alpha}) \Big) \right)  \\
         &+ \frac{(N+1)  L \ln(d)}{3} \exp_2 \left( \frac{1-\alpha}{2\alpha} E^{(\alpha)}(O) - \frac{\log(\chi)}{2\alpha}\right)  \\
         \leq& c(d,\alpha,L) N \exp_2 \left( \frac{1-\alpha}{2\alpha} \Big( E^{(\alpha)}(O_N) - \log({\chi}) \Big) \right). 
    \end{split}
\end{equation}
Here, in the second line we have used that $1-\alpha \leq 1$, $\chi^{-1/(2\alpha)}\le\chi^{-(1-\alpha)/(2\alpha)}$, $\sqrt{2(N+1)}\leq 2 N$, and $N+1\leq 2N$ to arrive at a simpler expression, and gathered the constants together,
\begin{equation}
    c(d,\alpha,L) := 2( \sqrt{\ln(d)  } (1-\alpha)^{\frac{1-\alpha}{ 2\alpha} } + 3^{-1} L \ln(d) ). \label{eq:const}
\end{equation}
This completes the proof of Thm.~\ref{thm:random}. One could, of course, not use $1-\alpha \leq 1$ and $N+1\leq 2N$ to achieve a stronger (but messier) bound. 

\begin{figure*}[t]
    \centering
    \includegraphics[width=0.8\linewidth]{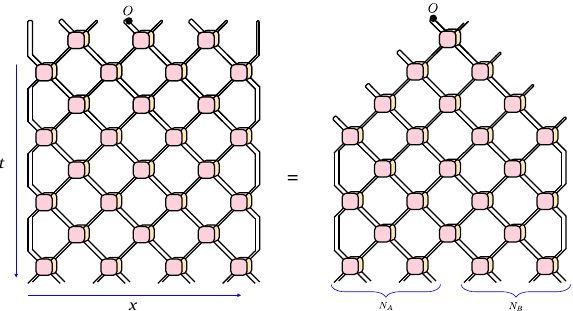}
    \caption{A vectorized operator $| O \rrangle$ under the action of a brickwork circuit in the folded picture. Red bricks (foreground) are the two-site unitary $U$, and yellow bricks (background) denote its conjugate $U^*$. Here, we show how the lightcone in such circuits is exact for an initially local operator $O$. The shown circuit is on $N=8$ qudits, and $t=8$ layers of the Floquet evolution. 
    }
    \label{fig:brickwork}
\end{figure*}

\section{Details of Numerics} \label{ap:numerics}
\appendixcontentsline{section}{C. Details of Numerics}{}
We study $N$-qubit spin chain models evolving under a discrete brickwork circuit~\cite{Nahum2017}, parametrized by a single two-site unitary. A single step of time evolution, $U$, consists of a layer of two-site unitaries applied to next-neighbor sites. When $t$ is odd, we take $U^{\otimes N/2-1}$ to be applied to sites $\{2,3\} $, $\{4,5\}, \dots$, with identity on the boundary sites $1$ and $N$ (we take $N$ to be even). When $t$ is even, $U^{\otimes N/2}$ is applied to the other next-neighbour pairings: sites $\{1,2\} $, $\{3,4\}, \dots$. Every two steps, therefore, defines a repeating Floquet evolution step. We show an example of a time-evolved operator in Fig.~\ref{fig:brickwork}. After $t$ steps, we denote the full evolution operator as $U_t$, and a time-evolved Heisenberg operator as $O_t:= U_t^\dagger O U_t$ for some initial Hermitian $O$.

For the two-body gate, we study the following models:
\begin{itemize}
    \item The XXZ model, 
    \begin{equation}
        U_{\mathrm{XXZ}}^{(2)}  = \exp\left( -i J (  \sigma_x \otimes \sigma_x + \sigma_y \otimes \sigma_y  + \Delta \sigma_z \otimes \sigma_z) )\right)
    \end{equation}
    with $J=1$ and $\Delta = 0.55$. This model is interacting integrable.
    \item The (kicked) Ising model (KIM)
    \begin{equation}
        U_{\mathrm{KIM}}^{(2)} = \exp\left( -i ( J \sigma_z \otimes \sigma_z + h_x( \id \ot \sigma_x +  \sigma_x \ot \id) + h_z( \id \ot \sigma_z +  \sigma_z \ot \id) ) \right).
    \end{equation}
    With $J=1,\, h_x=0.9045,\, h_z=0.8090 $ (corresponding to the data of Fig.~\ref{fig:numerics}), this model is non-integrable. For $h_z=0$, we recover the integrable transverse field Ising model (TFIM), and we also study this in Fig.~\ref{fig:numerics2} (with $J=1, \, h_x=0.55$). 
\end{itemize}
In both cases, taking $N$ to be even, we choose an initial operator $O$ to be a local $\sigma_z$ on the site $i=N/2 $. Then, after $t$ Floquet steps (where $t$ may be a half integer), the Heisenberg operator is $O_t =U^{\dagger}_t O U_t $. Consider an approximation to $O_t$ given by a truncation across the center of the spin chain:
\begin{equation}
    O_t = \sum_{i=1}^{d^{N}}  \lambda_i A_i \otimes B_i \to   \sum_{i=1}^{\chi}  \lambda_i A_i \otimes B_i =: \tilde{O}_{t,\chi}
\end{equation}
where clearly $\lambda_i$, $A_i$, and $B_i$ all depend on time. Notice that as the bipartition is across the half chain, in the worst case, we have $d^{2 N_A} = d^N$ non-zero Schmidt coefficients $\lambda_i$. The truncation is determined through a sharp cutoff: $\lambda_{\chi+1} < 0.02$ and $\lambda_{\chi} \geq 0.02$. Using exact diagonalization, for both of these models with $N=\{8,10,12 \}$, we compute the two relevant truncation errors appearing in this work: $\| \Delta O \|_\infty $ and $D^{-1/2} \| \Delta O \|_2 $ where $\Delta O := O_t - \tilde{O}_{t,\chi}$; see Fig.~\ref{fig:numerics2} (a). The results for the XXZ and KIM models for $N=12$ are shown together in the left panel of Fig.~\ref{fig:numerics}. We see that with increasing $N$, the truncation errors do not differ from each other by a significant amount (noting that $\| \Delta O \|_\infty$ could be greater than one in the worst case). Note that we are limited in system size as the vectorized operators live on a doubled Hilbert space of dimension $d^{2 N}$---so the computational cost of half-chain Schmidt decomposition of an $N=12 $ qubit operator is equivalent to a $N=24$ site qubit chain. For all models, we also see an approximate power-law growth in both truncation errors with time. 

In Fig.~\ref{fig:numerics2} we show the von Neumann LOE for the discussed models, again across the half-chain and for an initial $\sigma_z$. We see evidence of the logarithmic growth for the two integrable models (TFIM and XXZ), and a linear growth for the non-integrable model (KIM). This supports the conjecture that the scaling of LOE is a witness for the (non-)integrability of spin chains~\cite{Prosen2007}. Note that there are analytic results on the LOE for the TFIM: a time-evolved local $\sigma_x$ operator has a Majorana index of $2$, and thus an LOE trivially bounded by $2$~\cite{Prosen2007a}. For $\sigma_z$ (what we use here), it is also known analytically for the critical parameter $h_x=1$~\cite{Dubail_2017}, scaling logarithmically with time. Note also that the LOE of the Ising Model with longitudinal field (continuous version of the brickwork KIM studied here) is studied numerically in Ref.~\cite{Prosen2007}, and the continuous XXZ model in Ref.~\cite{Alba2021}.

We also compute the spectral norm of the matrices $A_i \otimes B_i$ from the operator Schmidt decomposition. The data for $N=12$ for the XXZ and KIM models is displayed in the right panel of Fig.~\ref{fig:numerics}, while the TFIM data is shown in Fig.~\ref{fig:numerics2} (c). The distribution is seen to concentrate around an $\mc{O}(1)$ value with increasing $N$. This data therefore serves as a justification for the bound $ \| A_i \otimes B_i \|_\infty \leq L$ for some $L=\mc{O}(1)$ assumed in the random matrix theory model studied in Thm.~\ref{thm:random}. Interestingly, we find that for the TFIM, $\| A_i \otimes B_i \|_\infty \approx 1$ for all $i$. This likely stems from its Majorana structure, and it is an interesting question if this behavior can be derived analytically, perhaps using the results of Ref.~\cite{Dubail_2017}. 

\begin{figure*}[b]
    \centering
    \includegraphics[width=0.8\linewidth]{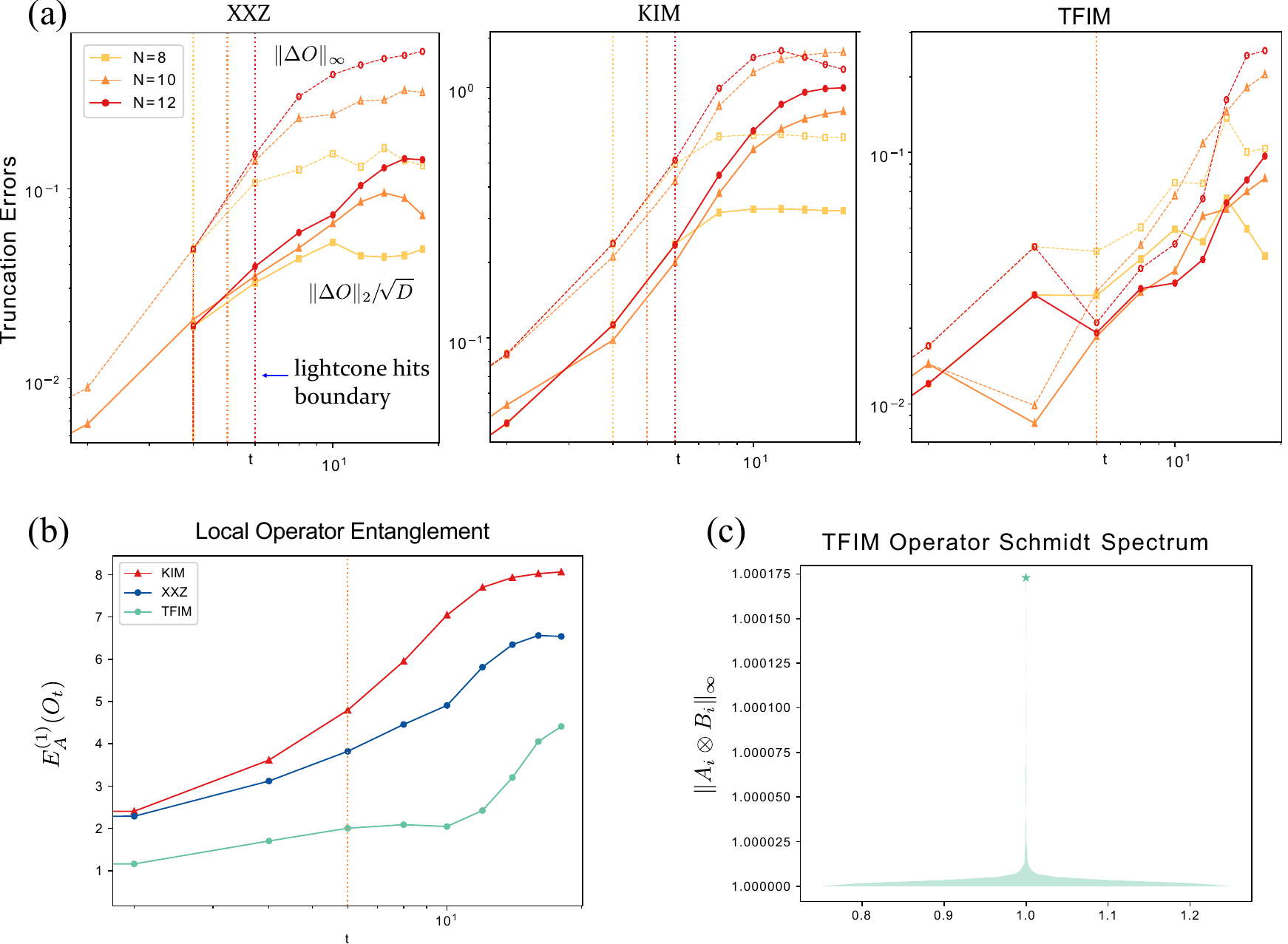}
    \caption{Additional numerical results. (a) We first compute the truncation errors, as studied in the left panel of Fig.~\ref{fig:numerics}, but for different system sizes $N=\{8,10,12\}$ and also including the integrable transverse field Ising model (TFIM). For different system sizes, the light cone of the initially local operator (on the center spin) hits the boundary at different depths, which are marked by suitably colored, vertical dotted lines. (b) We compute the exact von Neumann LOE of the half chain of the studied models for $N=12$ up to $t=20$ layers of evolution. (c) We provide data on the spectral norm distribution of $A_i \otimes B_i$ for the half-chain bipartition for the TFIM, complementary to the results on the right panel of Fig.~\ref{fig:numerics}.
    }
    \label{fig:numerics2}
\end{figure*}

\end{document}